\documentclass[submission,copyright,creativecommons]{eptcs}
\providecommand{\event}{WPTE 2018} 

\usepackage[all]{xy}

\DeclareMathAlphabet{\mathcal}{OMS}{cmsy}{m}{n}

\usepackage{theorem}

\newcommand{\eop}{\hspace*{\fill}$\Box$}
\def\qed{\eop}
\newtheorem{definition}{Definition}[section]
\newtheorem{lemma}[definition]{Lemma}
\newtheorem{theorem}[definition]{Theorem}
{\theorembodyfont{\rmfamily}
  \newtheorem{example}[definition]{\it Example}
  \newtheorem{proof}{\it Proof.}
}

\usepackage{amsmath}
\usepackage{amssymb}
\usepackage{stmaryrd}

\def\cC{\mathcal{C}}
\def\cD{\mathcal{D}}
\def\cF{\mathcal{F}}
\def\cG{\mathcal{G}}
\def\cN{\mathcal{N}}
\def\cP{\mathcal{P}}
\def\cR{\mathcal{R}}
\def\cT{\mathcal{T}}
\def\cV{\mathcal{V}}
\def\Root{\mathit{root}}
\def\Arity{\mathit{arity}}
\def\Var{{\mathcal{V}\mathit{ar}}}
\def\Terms{\cT}
\def\Pos{{\mathcal{P}\mathit{os}}}
\def\Dom{{\mathcal{D}\mathit{om}}}
\def\Range{{\mathcal{R}\mathit{an}}}
\def\VRange{{\mathcal{VR}\mathit{an}}}
\def\Hole{\Box}
\def\Subst{\mathit{Subst}}

\def\Pcomp{\mathrel{\Uparrow}}
\def\Fail{\mathit{fail}}
\def\Scomp{\mathrel{\cdot}}
\def\ScompSym{\mathrel{\bullet}}
\def\mgu{\mathit{mgu}}

\def\CALL{\textsc{rec}}
\def\EmptysetSym{\varnothing}

\def\tto{\twoheadrightarrow}
\def\Condition#1#2{#1 \tto #2}

\def\Rule#1#2{#1 \to #2}
\def\SSGRule#1#2{#1 \to #2}

\def\CRule#1#2#3{\Rule{#1}{#2} \Leftarrow #3}
\def\CCPair#1#2#3{\langle #1, #2 \rangle \Leftarrow #3}

\def\ID{\mathit{id}}
\def\symb#1{\mathsf{#1}}
\def\metasymb#1{\mathit{#1}}
\def\var#1{\mathit{#1}}

\def\Tend{\Terms(\{\top,\Eand\})}
\def\Eqn#1#2{#1 \approx #2}
\def\BraCondition#1#2{(\Condition{#1}{#2})}
\def\Eqns{\cdots}
\newcommand{\Eand}[1][]{\mathrel{\&}}
\def\ApplySubst#1#2{{#1(#2)}}
\def\applySubst#1#2{{#1 #2}}
\def\Lang{L}

\newcommand{\Semantics}[3][]{[\![\,#2\,]\!]}
\newcommand{\RenameFresh}[2]{\mathit{fresh}_{#2}(#1)}

\def\SubstSet#1#2{[\![\Lang(#1, #2)]\!]}

\newcommand{\cto}{\stackrel{\mathit{c}}{\to}\hspace{-3pt}{}}
\newcommand{\ileadsto}{\stackrel{\mathit{i}}{\leadsto}\hspace{-3pt}{}}

\newcommand{\Coding}[2]{\left[\,#1,~#2\,\right]}
\newcommand{\Pairsymb}[2]{\symb{#1 #2}}
\newcommand{\NT}[2]{{#1#2}}
\newcommand{\CodingNT}[3]{{\langle #1,~#2 \rangle}_{#3}}
\newcommand{\Patterns}[1]{{\mathit{Patterns}}(#1)}
\newcommand{\Vars}[1]{{\mathit{Vars}}(#1)}
\def\cRgcd{\cR_1}

\def\cGgcd{\cG_1}
\def\cPgcd{\cP_1}
\def\cGex{\cG_2}

\def\cGgcddomain{\cG_4}
\def\cGcex{\cG_5}

\title{On Transforming Narrowing Trees into Regular Tree Grammars Generating Ranges of Substitutions%
\thanks{This work was partially supported by JSPS KAKENHI Grant Number JP17H01722.} 
}

\author{Naoki Nishida
\institute{Graduate School of Informatics\\
Nagoya University\\
Nagoya, Japan}
\email{nishida@i.nagoya-u.ac.jp}
\and
Yuya Maeda
\institute{Graduate School of Informatics\\
Nagoya University\\
Nagoya, Japan}
\email{yuya@trs.css.i.nagoya-u.ac.jp}
}
\def\titlerunning{On Transforming Narrowing Trees into RTGs Generating Ranges of Substitutions}
\def\authorrunning{N. Nishida and Y. Maeda}

\begin{document}
\maketitle

\begin{abstract}
The grammar representation of a narrowing tree for a syntactically deterministic conditional term rewriting system and a pair of terms is a regular tree grammar that generates expressions for substitutions obtained by all possible innermost-narrowing derivations that start with the pair and end with particular non-narrowable terms.
In this paper, under a certain syntactic condition, we show a transformation of the grammar representation of a narrowing tree into another regular tree grammar that overapproximately generates the ranges of ground substitutions generated by the grammar representation.
In our previous work, such a transformation is restricted to the ranges w.r.t.\ a given single variable, and thus, the usefulness is limited.
We extend the previous transformation by representing the range of a ground substitution as a tuple of terms, which is obtained by the coding for finite trees.
We show a precise definition of the transformation and prove that the language of the transformed regular tree grammar is an overapproximation of the ranges of ground substitutions generated by the grammar representation.
We leave an experiment to evaluate the usefulness of the transformation as future work.
\end{abstract}

\section{Introduction}
\label{sec:intro}

\emph{Conditional term rewriting}~\cite[Chapter~7]{Ohl02} is known to be more complicated
than unconditional term rewriting in the sense of analyzing properties,
e.g., \emph{operational termination}~\cite{LMM05},
\emph{confluence}~\cite{SMI95}, and \emph{reachability}~\cite{FG03}.
A popular approach to the analysis of conditional rewriting is to transform a conditional term rewriting system (a CTRS, for short) into an unconditional term rewriting
system (a TRS, for short) that is in general an overapproximation of the CTRS in
terms of reduction. 
This approach enables us to use existing techniques for the analysis of TRSs.
For example, a CTRS is operationally terminating if the \emph{unraveled} TRS~\cite{Mar96,Ohl02}
 is terminating~\cite{DLMMU04}.
To prove termination of the
unraveled TRS, we can use many techniques for proving termination of TRSs (cf.~\cite{Ohl02}).
On the other hand, it is not so easy to analyze \emph{reachability} which is relevant to, e.g., \emph{(in)feasibility} of conditions. 

Let us consider to prove confluence of the following \emph{syntactically deterministic 3}-CTRS~\cite[Example~7.1.5]{Ohl02} defining the \emph{gcd} operator over the natural numbers represented by $\symb{0}$ and $\symb{s}$:
\[
	\cRgcd =
	\left\{
	\begin{array}{r@{\>}c@{\>}l@{~~~~~~~~}r@{\>}c@{\>}l}
	\Rule{x < \symb{0} &}{& \symb{false}}, &
	\Rule{\symb{0} - \symb{s}(y) &}{& \symb{0}}, \\
	\Rule{\symb{0} < \symb{s}(y) &}{& \symb{true}}, &
	\Rule{x - \symb{0} &}{& x}, \\
	\Rule{\symb{s}(x) < \symb{s}(y) &}{& x < y}, &
	\Rule{\symb{s}(x) - \symb{s}(y) &}{& x - y}, \\
	\Rule{\symb{gcd}(x,x) &}{& x}, \\
	\Rule{\symb{gcd}(\symb{s}(x),\symb{0}) &}{& \symb{s}(x)}, &
	\CRule{\symb{gcd}(\symb{s}(x),\symb{s}(y)) &}{& \symb{gcd}(x - y, \symb{s}(y))}{\Condition{y < x}{\symb{true}}}, \\
	\Rule{\symb{gcd}(\symb{0},\symb{s}(y)) &}{& \symb{s}(y)}, &
	\CRule{\symb{gcd}(\symb{s}(x),\symb{s}(y)) &}{& \symb{gcd}(\symb{s}(x), y - x)}{\Condition{x < y}{\symb{true}}} \\
	\end{array}
	\right\}
\]	
A transformational approach in~\cite{GNG13iwc,GN14wpte} does not succeed in proving confluence of $\cRgcd$.
On the other hand, a direct approach to reachability analysis to prove \emph{infeasibility} of the \emph{conditional critical pairs} (i.e., non-existence of substitutions satisfying conditions), which is implemented in some confluence provers, does not prove confluence of $\cRgcd$ well, either.
Let us consider the critical pairs of $\cRgcd$:
\[
	\begin{array}{@{}c@{~}c@{~~}c@{~}l@{}}
	\CCPair{& \symb{s}(x)}{& \symb{gcd}(x - x, \symb{s}(x)) &}{\Condition{x < x}{\symb{true}}},
	\\
	\CCPair{& \symb{gcd}(x - x, \symb{s}(x))}{& \symb{s}(x) &}{\Condition{x < x}{\symb{true}}},
	\\
	\CCPair{& \symb{s}(x)}{& \symb{gcd}(\symb{s}(x), x - x) &}{\Condition{x < x}{\symb{true}}},
	\\
	\CCPair{& \symb{gcd}(\symb{s}(x), x - x)}{& \symb{s}(x) &}{\Condition{x < x}{\symb{true}}},
	\\
	\CCPair{& \symb{gcd}(x - y, \symb{s}(y))}{& \symb{gcd}(\symb{s}(x), y - x) &}{\Condition{x < y}{\symb{true}},~\Condition{y < x}{\symb{true}}},
	\\
	\CCPair{& \symb{gcd}(\symb{s}(x), y - x)}{& \symb{gcd}(x - y, \symb{s}(y)) &}{\Condition{x < y}{\symb{true}},~\Condition{y < x}{\symb{true}}}
	\\
	\end{array}
\]
Note that the above critical pairs are symmetric because they are caused by overlaps at the root position only.
An operationally terminating CTRS is confluent if all critical pairs of the CTRS are infeasible (cf.~\cite{BN98,DOS88}).
Operational termination of $\cRgcd$ can be proved by, e.g., \textsf{AProVE}~\cite{AProVE}.
To prove infeasibility of the critical pairs above, it suffices to show 
both (i) non-existence of terms $t$ such that $t < t \mathrel{\to^*_{\cRgcd}} \symb{true}$, and (ii) non-existence of terms $t_1,t_2$ such that $t_1 < t_2 \mathrel{\to^*_{\cRgcd}} \symb{true}$ and $t_2 < t_1 \mathrel{\to^*_{\cRgcd}} \symb{true}$.
Thanks to the meaning of $<$, it would be easy for a human to notice that such terms $t,t_1,t_2$ do not exist.
However, it is not so easy to mechanize a way to show non-existence of $t,t_1,t_2$.
In fact, confluence provers for CTRSs,
\textsf{ConCon}~\cite{concon-2014}, \textsf{CO3}~\cite{co3-2015}, and \textsf{CoScart}~\cite{coscart-2015}, based on e.g., transformations of CTRSs into TRSs or reachability analysis for infeasibility of conditional critical pairs, failed to prove confluence of $\cRgcd$ (see 
Confluence Competition 2016,
2017,
and 2018,%
\footnote{ \url{http://cops.uibk.ac.at/results/?y=2018&c=CTRS}}
\texttt{327.trs}).
In addition, a \emph{semantic approach} in~\cite{Luc17iwc,Luc17lopstr} cannot prove confluence of $\cRgcd$ using \textsf{AGES}~\cite{GLR16}, a tool for generating logical models of order-sorted first-order theories---non-existence of $t_1,t_2$ above cannot be proved via its web interface with default parameters.
\textsf{Timbuk 3.2}~\cite{Timbuk01}, which is based on tree automata techniques~\cite{GR10}, cannot prove infeasibility of $\Condition{x < y}{\symb{true}},~\Condition{y < x}{\symb{true}}$ w.r.t.\ the rules for $<$ under the default use.

The non-existence of a term $t$ with $t < t \mathrel{\to^*_{\cRgcd}} \symb{true}$  can be reduced to the non-existence of substitutions $\theta$ such that $x < x \mathrel{\leadsto^*_{\theta,\cRgcd}} \symb{true}$, where $\leadsto$ denotes the \emph{narrowing} step~\cite{Hul80}---%
for example, $x < y \mathrel{\leadsto_{\{x \mapsto \symb{0}, ~ y \mapsto \symb{s}(y')\},\cRgcd}} \symb{true}$.
In addition, the non-existence of such substitutions can be reduced to the emptiness of the set of the substitutions, i.e., the emptiness of $\{ \theta \mid x < x \mathrel{\leadsto^*_{\theta,\cRgcd}} \symb{true} \}$.
From this viewpoint, for a pair of terms, the enumeration of substitutions obtained by narrowing would be useful in analyzing rewriting that starts with instances of the pair.
To analyze sets of substitutions derived by \emph{innermost} narrowing, \emph{narrowing trees}~\cite{NV13} are useful. 
For example, infeasibility of conditional critical pairs of some normal 1-CTRS can be proved by using the \emph{grammar representation} of a narrowing tree~\cite{NM18fscd}.
Simplification of the grammar representation implies the non-existence of substitutions satisfying the conditional part of a critical pair.
However, there are some examples (shown later) for which the simplification method in~\cite{NM18fscd} does not succeed in converting grammar representations to those explicitly representing the empty set.

In this paper, under a certain syntactic condition, we show a transformation of the grammar representation of a narrowing tree into a \emph{regular tree grammar}~\cite{TATA} (an RTG, for short) that overapproximately generates the ranges of ground substitutions generated by the grammar representation.
The aim of the transformation is to simplify grammar representations as much as possible together with the existing one in~\cite{NM18fscd}.

Let $\cR$ be a \emph{syntactically deterministic} 3-CTRS (a 3-SDCTRS, for short) that is a constructor system, $s$ a \emph{basic} term, and $t$ a constructor term, where basic terms are of the form $\metasymb{f}(u_1,\ldots,u_n)$ with a defined symbol $\metasymb{f}$ and constructor terms $u_1,\ldots,u_n$.
A \emph{narrowing tree}~\cite{NV13,NM18fscd} of $\cR$ with the root pair $\Condition{s}{t}$ is a finite representation that defines the set of substitutions $\theta$ such that 
the pair $\Condition{s}{t}$ narrows to a particular ground term $u_\top$ consisting of a special binary symbol $\Eand$ and a special constant $\top$ by \emph{innermost} narrowing ${\ileadsto_\cR}$ with a substitution $\theta$ (i.e., $\BraCondition{s}{t} \mathrel{\ileadsto^*_{\theta,\cR}} u_\top$ and thus $\applySubst{\theta}{s} \mathrel{\cto^*_\cR} \applySubst{\theta}{t}$).
Note that $\Condition{}{}$ is considered a binary symbol, $\Rule{\BraCondition{x}{x}}{\top}$ is assumed to be implicitly included in $\cR$, and $\cto_\cR$ denotes the \emph{constructor-based rewriting} step which applies rewrite rules to basic terms.
Such a narrowing tree can be the enumeration of substitutions obtained by innermost narrowing of $\cR$ to ground terms consisting of $\Eand$ and $\top$.
The idea of narrowing trees has been extended to finite representations of SLD trees for logic programs~\cite{NV15}.

Using narrowing trees, it is easy to see that there is no substitution $\theta$ such that $x < x \mathrel{\ileadsto^*_{\theta,\cRgcd}} \symb{true}$, and hence the above four critical pairs with $\Condition{x < x}{\symb{true}}$ are infeasible.
Let us now consider to prove infeasibility of $\Condition{x < y}{\symb{true}},~\Condition{y < x}{\symb{true}}$.
A narrowing tree for $\Condition{x < y}{\symb{true}} \Eand \Condition{y < x}{\symb{true}}$ can be represented by the following \emph{grammar representation}~\cite{NV13,NM18fscd} that can be considered an RTG (see Section~\ref{sec:grammar_representation}):
\begin{equation}
\label{eqn:gcd}
	\begin{array}{@{}r@{\>\>}r@{\>\>}l@{}}
	\SSGRule{\Gamma_{\Condition{x < y}{\symb{true}} \Eand \Condition{y < x}{\symb{true}}} &}{& \Gamma_{\Condition{x < y}{\symb{true}}} \Eand \Gamma_{\Condition{y < x}{\symb{true}}}} 
	\\[3pt]
	\SSGRule{\Gamma_{\Condition{x < y}{\symb{true}}} &}{& 
	\{ x \mapsto \symb{0}, ~ y \mapsto \symb{s}(y_2) \}
	\\
	&
	\mid
	&
	\CALL( \Gamma_{\Condition{x < y}{\symb{true}}} , \{ x_3 \mapsto x, ~ y_3 \mapsto y \} ) \ScompSym \{ x \mapsto \symb{s}(x_3), ~ y \mapsto \symb{s}(y_3) \}
	}
	\\[3pt]
	\SSGRule{\Gamma_{\Condition{y < x}{\symb{true}}} &}{& \CALL(\Gamma_{\Condition{x < y}{\symb{true}}},\{ x\mapsto y, ~ y \mapsto x \})}
	\\
 \end{array}
\end{equation}
We denote by $\cGgcd$ the RTG
with the initial non-terminal $\Gamma_{\Condition{x < y}{\symb{true}} \Eand \Condition{y < x}{\symb{true}}}$, the other non-terminals $\Gamma_{\Condition{x < y}{\symb{true}}}, \Gamma_{\Condition{y < x}{\symb{true}}}$, and the above production rules.
We also denote by $\cPgcd$ the set of the above production rules, i.e., (\ref{eqn:gcd}).
Substitutions are considered constants, and the RTG generates terms over $\Eand$, $\EmptysetSym$, $\ScompSym$, $\CALL$, and substitutions.
The binary symbols $\ScompSym$ and $\Eand$ are interpreted by standard composition and \emph{parallel composition}~\cite{HR89,Pal90}, 
respectively.
Parallel composition $\Pcomp$ of two substitutions returns a most general unifier of the substitutions if the substitutions are unifiable (see Definition~\ref{def:pcomp}).
For example, $\{y' \mapsto \symb{a}, ~ y \mapsto \symb{a}\} \Pcomp \{ y' \mapsto y\}$ returns $\{y' \mapsto \symb{a}, ~ y \mapsto \symb{a}\}$ and $\{y' \mapsto \symb{a}, ~ y \mapsto \symb{b}\} \Pcomp \{ y' \mapsto y\}$ fails.
The symbol $\CALL$ is used for recursion, 
which is interpreted as standard composition of a renaming and a substitution recursively generated.
To simplify the discussion in the remainder of this section, following the meaning of the operators, we simplify the rules of $\Gamma_{\Condition{x < y}{\symb{true}}}$ and $\Gamma_{\Condition{y < x}{\symb{true}}}$ as follows:
\begin{equation}
\label{eqn:gcd-simple}
	\begin{array}{@{}r@{\>\>}r@{\>\>}l@{}}
	\SSGRule{\Gamma_{\Condition{x < y}{\symb{true}}} &}{& 
	\{ x \mapsto \symb{0}, ~ y \mapsto \symb{s}(y_2) \}
	\mid
	\Gamma_{\Condition{x < y}{\symb{true}}} \ScompSym \{ x \mapsto \symb{s}(x), ~ y \mapsto \symb{s}(y) \}
	}
	\\[3pt]
	\SSGRule{\Gamma_{\Condition{y < x}{\symb{true}}} &}{& \Gamma_{\Condition{x < y}{\symb{true}}} \ScompSym \{ x\mapsto y, ~ y \mapsto x \}}
	\\
 \end{array}
\end{equation}

In our previous work~\cite{NM18fscd}, to show the emptiness of the set of substitutions generated from e.g., $\Gamma_{\Condition{x < y}{\symb{true}}} \Eand \Gamma_{\Condition{y < x}{\symb{true}}}$, we transform the grammar representation to an RTG that overapproximately generates the ranges of ground substitutions w.r.t.\ a single variable.
For example, for $x$, the production rules of (\ref{eqn:gcd-simple}) is transformed into the following ones:
\[
\begin{array}{@{}l@{\quad\quad}l@{}}
		\SSGRule{\Gamma_{\Condition{x < y}{\symb{true}}}^x}{
		\symb{0} 
		\mid
		\symb{s}(\Gamma_{\Condition{x < y}{\symb{true}}}^x)
		}
		&
		\SSGRule{\Gamma_{\Condition{y < x}{\symb{true}}}^x}{ \Gamma_{\Condition{x < y}{\symb{true}}}^y
		}
\\[3pt]
		\SSGRule{\Gamma_{\Condition{x < y}{\symb{true}}}^y}{
		\symb{s}(A) 
		\mid
		\symb{s}(\Gamma_{\Condition{x < y}{\symb{true}}}^y)
		}
		&
		\SSGRule{A}{\symb{0} \mid \symb{s}(A) \mid \symb{true} \mid \symb{false}
		}
\end{array}
\]
Note that non-terminal $A$ generates arbitrary ground constructor terms.
Since we focus on $x$ only, non-terminals $\Gamma_{\Condition{x < y}{\symb{true}}}^x$ and $\Gamma_{\Condition{y < x}{\symb{true}}}^x$ generate $\{ \symb{s}^n(a) \mid n \geq 0, ~ a \in \{\symb{0},\symb{true},\symb{false}\} \}$ and $\{ \symb{s}^n(a) \mid n > 0, ~ a \in \{\symb{0},\symb{true},\symb{false}\} \}$, respectively, and we cannot prove that there is no substitution generated from $\Gamma_{\Condition{x < y}{\symb{true}}} \Eand \Gamma_{\Condition{y < x}{\symb{true}}}$.

In this paper, we aim at showing that there is no substitution generated by (\ref{eqn:gcd-simple}) from the initial non-terminal $\Gamma_{\Condition{x < y}{\symb{true}} \Eand \Condition{y < x}{\symb{true}}}$, i.e., 
showing that $\Lang(\cGgcd,\Gamma_{\Condition{x < y}{\symb{true}}}) \cap \Lang(\cGgcd,\Gamma_{\Condition{y < x}{\symb{true}}}) = \emptyset$.
To this end, under a certain syntactic condition, we show a transformation of the grammar representation of a narrowing tree into an RTG that overapproximately generates the ranges of ground substitutions generated by the grammar representation (Section~\ref{sec:transformation}).
More precisely, using the idea of \emph{coding} for tuples of ground terms~\cite[Section~3.2.1]{TATA} (see Figure~\ref{fig:ideaOfCoding}), we extend a transformation in~\cite{NM18fscd} w.r.t.\ a single variable to two variables.
It is straightforward to further extend the transformation to three or more variables.
We do not explain how to, given a constructor 3-SDCTRS, construct (the grammar representation of) a narrowing tree, and concentrate on how to transform a  grammar representation into an RTG that generates the ranges of ground substitutions generated by the grammar representation.

\begin{figure}[t]
\vspace{-5pt}
\[
\Coding{%
\raisebox{7mm}{%
\xymatrix@R=6pt@C=3pt{
 & \symb{f} \ar@{-}[ld]\ar@{-}[rd] \\
 \symb{g} \ar@{-}[d] & & \symb{g} \ar@{-}[d] \\
 \symb{a} & & \symb{a} \\
}}
~~}{%
\raisebox{7mm}{%
\xymatrix@R=6pt@C=3pt{
& & \symb{f} \ar@{-}[ld]\ar@{-}[rd] \\
& \symb{f} \ar@{-}[dl] \ar@{-}[dr] & & \symb{a} \\
\symb{a} & & \symb{a} \\
}}
}
\quad = \quad
\raisebox{8mm}{%
\xymatrix@R=6pt@C=3pt{
& & \Pairsymb{f}{f} \ar@{-}[ld]\ar@{-}[rd] \\
& \Pairsymb{g}{f} \ar@{-}[dl] \ar@{-}[dr] & & \Pairsymb{g}{a} \ar@{-}[d] \\
\Pairsymb{a}{a} & & \Pairsymb{\bot}{a} & \Pairsymb{a}{\bot} \\
}}
\]
\vspace{-10pt}
\caption{the coding of $\symb{f}(\symb{g}(\symb{a}),\symb{g}(\symb{a}))$ and $\symb{f}(\symb{f}(\symb{a},\symb{a}),\symb{a})$.}
\label{fig:ideaOfCoding}	
\end{figure}

\paragraph{Outline of Our Approach}
Using the rules of (\ref{eqn:gcd-simple}), we briefly illustrate the outline of the transformation.
Roughly speaking, we apply the coding for tuples of terms to the range of substitutions, e.g., $\symb{0}$ and $\symb{s}(y_2)$ for $\{ x \mapsto \symb{0}, ~ y \mapsto \symb{s}(y_2) \}$.
%
The rules for $\Gamma_{\Condition{x < y}{\symb{true}}}$ are transformed into
\[
\SSGRule{\Gamma_{\Condition{x < y}{\symb{true}}}^{(x,y)}}{ \Pairsymb{0}{s}(\NT{\bot}{A}) }
\qquad
\SSGRule{\Gamma_{\Condition{x < y}{\symb{true}}}^{(x,y)}}{ \Pairsymb{s}{s}(\Gamma_{\Condition{x < y}{\symb{true}}}^{(x,y)}) }.
\]
where the non-terminal $\NT{\bot}{A}$ generates ground terms obtained by applying the coding to $\bot$ and ground constructor terms.
The coding of $\symb{s}(x)$ and $\symb{s}(y)$ is $\Pairsymb{s}{s}(\NT{x}{y})$.
Variables $x,y$ are instantiated by substitutions generated from $\Gamma_{\Condition{x < y}{\symb{true}}}$, and hence we replaced $\NT{x}{y}$ by $\Gamma_{\Condition{x < y}{\symb{true}}}^{(x,y)}$.
The rule for $\Gamma_{\Condition{y < x}{\symb{true}}}$ is transformed into
\[
\SSGRule{\Gamma_{\Condition{y < x}{\symb{true}}}^{(x,y)}}{ \Gamma_{\Condition{x < y}{\symb{true}}}^{(y,x)} }.
\]
Since $x,y$ are swapped by $\{ x \mapsto y, ~ y \mapsto x \}$, we generate a new non-terminal $\Gamma_{\Condition{x < y}{\symb{true}}}^{(y,x)}$ and its rules as well as the above rules:
\[
\SSGRule{\Gamma_{\Condition{x < y}{\symb{true}}}^{(y,x)}}{ \Pairsymb{s}{0}(\NT{A}{\bot}) }
\qquad
\SSGRule{\Gamma_{\Condition{x < y}{\symb{true}}}^{(y,x)}}{ \Pairsymb{s}{s}(\Gamma_{\Condition{x < y}{\symb{true}}}^{(y,x)}) }.
\]
where the non-terminal $\NT{A}{\bot}$ generates ground terms obtained by applying the coding to ground constructor terms and $\bot$.
Every ground term generated from $\Gamma_{\Condition{x < y}{\symb{true}}}^{(x,y)}$ contains $\Pairsymb{0}{s}$, and every ground term generated from $\Gamma_{\Condition{y < x}{\symb{true}}}^{(x,y)}$ contains $\Pairsymb{s}{0}$.
Neither $\Pairsymb{0}{s}$ nor $\Pairsymb{s}{0}$ is shared by the languages of $\Gamma_{\Condition{x < y}{\symb{true}}}^{(x,y)}$ and $\Gamma_{\Condition{y < x}{\symb{true}}}^{(x,y)}$, and hence there is no substitution which corresponds to an expression generated from $\Gamma_{\Condition{x < y}{\symb{true}} \Eand \Condition{y < x}{\symb{true}}}$.
For this reason, we can transform $\Gamma_{\Condition{x < y}{\symb{true}} \Eand \Condition{y < x}{\symb{true}}}$ of~(\ref{eqn:gcd}) into
\[
\SSGRule{\Gamma_{\Condition{x < y}{\symb{true}} \Eand \Condition{y < x}{\symb{true}}}}{\EmptysetSym}
\]
which means that there exist no constructor substitution $\theta$ satisfying the condition $\Condition{x < y}{\symb{true}} \Eand \Condition{y < x}{\symb{true}}$ under the constructor-based rewriting.

\medskip
One may think that tuples of terms are enough for our goal.
However, substitutions are generated by standard compositions, and tuples makes us introduce composition of tuples.
For example, the range of $\sigma=\{x \mapsto \symb{f}(x',\symb{g}(\symb{a})), ~ y \mapsto \symb{f}(y',\symb{a}) \}$ is represented as a tuple
$\symb{tup_2}(\symb{f}(x',\symb{g}(\symb{a})), \symb{f}(y',\symb{a}))$,
where $\symb{tup_2}$ is a binary symbol for tuples of two terms.
To apply $\theta=\{ x' \mapsto \symb{g}(\symb{a}), ~ y' \mapsto \symb{f}(\symb{a},\symb{a}) \}$ to the tuple, we reconstruct a tuple from $\symb{tup_2}(\symb{f}(x',\symb{g}(\symb{a})), \symb{f}(y',\symb{a}))$ and $\theta$.
On the other hand, the coding of terms makes us avoid the reconstruction and use standard composition of substitutions to compute the range of composed substitution. For example, $\sigma$ and $\theta$ can be represented by $\{ {x}{y} \mapsto \Pairsymb{f}{f}({x'}{y'}, \Pairsymb{g}{a}(\Pairsymb{a}{\bot})) \}$ and $\{ {x'}{y'} \mapsto \Pairsymb{g}{f}(\Pairsymb{a}{a}, \Pairsymb{\bot}{a}) \}$, respectively, where both ${x}{y}$ and ${x'}{y'}$ are considered single variables.

Using the rules for $\Gamma_{\Condition{x < y}{\symb{true}}}$ of~(\ref{eqn:gcd-simple}), we further show that the weakness of the above approach of using tuples.
Let us try to transform the rules of $\Gamma_{\Condition{x < y}{\symb{true}}}$ into an RTG that generates $\{ \symb{tup_2}(\symb{s}^m(\symb{0}),\symb{s}^n(a)) \mid 0 \leq m < n, ~ a \in \{\symb{0},\symb{true},\symb{false}\} \}$.
The first rule $\SSGRule{\Gamma_{\Condition{x < y}{\symb{true}}}}{\{ x \mapsto \symb{0}, ~ y \mapsto \symb{s}(y_2) \}}$ is transformed into $\SSGRule{\Gamma_{\Condition{x < y}{\symb{true}}}^{(x,y)}}{ \symb{tup_2}(\symb{0},\symb{s}(A)) }$ with the rules of $A$ above.
The second rule $\SSGRule{\Gamma_{\Condition{x < y}{\symb{true}}}}{ \Gamma_{\Condition{x < y}{\symb{true}}} \ScompSym \{ x \mapsto \symb{s}(x), ~ y \mapsto \symb{s}(y) \} }$ is transformed into $\SSGRule{\Gamma_{\Condition{x < y}{\symb{true}}}^{(x,y)}}{ \symb{tup_2}(\symb{s}(\Gamma_{\Condition{x < y}{\symb{true}}}^x),\symb{s}(\Gamma_{\Condition{x < y}{\symb{true}}}^y)) }$ with the rules of $\Gamma_{\Condition{x < y}{\symb{true}}}^x$ and $\Gamma_{\Condition{x < y}{\symb{true}}}^y$ above.
These rules generates not only terms in $\{ \symb{tup_2}(\symb{s}^m(\symb{0}),\symb{s}^n(a)) \mid 0 \leq m < n, ~ a \in \{\symb{0},\symb{true},\symb{false}\} \}$ but also other terms, e.g., $\symb{tup_2}(\symb{s}(\symb{0}),\symb{s}(\symb{0}))$.
The term $\symb{tup_2}(\symb{s}(\symb{0}),\symb{s}(\symb{0}))$ should not be generated because the term can be a common element generated by $\Gamma_{\Condition{x < y}{\symb{true}}}^{(x,y)}$ and $\Gamma_{\Condition{y < x}{\symb{true}}}^{(x,y)}$ and we cannot prove $\Gamma_{\Condition{x < y}{\symb{true}}} \Eand \Gamma_{\Condition{y < x}{\symb{true}}}$ does not generate any substitution. 


\section{Preliminaries}
\label{sec:preliminaries}

In this section, we recall basic notions and notations of term
rewriting~\cite{BN98,Ohl02} and regular tree grammars~\cite{TATA}. 
Familiarity with basic notions on term rewriting~\cite{BN98,Ohl02} is assumed.

\subsection{Terms and Substitutions}

Throughout the paper, we use $\cV$ as a countably infinite set of
\emph{variables}. 
Let $\cF$ be a \emph{signature}, a finite set of \emph{function symbols} $\metasymb{f}$
each of which has its own fixed arity, denoted by $\Arity(\metasymb{f})$.
We often write $\metasymb{f}/n \in \cF$ instead of ``an $n$-ary symbol $\metasymb{f} \in \cF$'', and so on.
The set of \emph{terms} over $\cF$ and $V$ ($\subseteq \cV$) is denoted
by $\Terms(\cF,V)$, and $\Terms(\cF,\emptyset)$, the set of \emph{ground} terms, is abbreviated to $\Terms(\cF)$.
The set of variables appearing in any of terms
$t_1,\ldots,t_n$ is denoted by $\Var(t_1,\ldots,t_n)$.
We denote the set of positions of a term $t$ by $\Pos(t)$.
For a term $t$ and a position $p$ of $t$, the \emph{subterm} of $t$ at
$p$ is denoted by $t|_p$.
The function symbol at the \emph{root} position $\varepsilon$ of a term
$t$ is denoted by $\Root(t)$.
Given terms $s,t$ and a position $p$ of $s$, we denote by $s[t]_p$ the term obtained from $s$ by replacing the subterm $s|_p$ at $p$ by $t$.

A \emph{substitution} $\sigma$ is a mapping from variables to terms such that the number of variables $x$ with $\sigma(x)\ne x$ is finite, and is naturally extended over terms.
The \emph{domain} and \emph{range} of $\sigma$ are denoted by
$\Dom(\sigma)$ and $\Range(\sigma)$, respectively.
The set of variables in $\Range(\sigma)$ is denoted by $\VRange(\sigma)$: 
$\VRange(\sigma)=\bigcup_{x\in\Dom(\sigma)} \Var(\applySubst{\sigma}{x})$.
We may denote $\sigma$ by $\{ x_1 \mapsto t_1, ~ \ldots, ~ x_n \mapsto
t_n \}$ if $\Dom(\sigma) = \{x_1,\ldots,x_n\}$ and $\sigma(x_i) = t_i$
for all $1 \leq i \leq n$.
The \emph{identity} substitution is denoted by $\ID$.
The set of substitutions that range over a signature $\cF$ and a set $V$ of variables is denoted by $\Subst(\cF,V)$: 
$\Subst(\cF,V) = \{ \sigma \mid \mbox{$\sigma$ is a substitution}, ~ \Range(\sigma) \subseteq \Terms(\cF,V) \}$.
The application of a substitution $\sigma$ to a term $t$ is abbreviated to $\applySubst{\sigma}{t}$, and $\applySubst{\sigma}{t}$ is called
an \emph{instance} of $t$.
Given a set $V$ of variables, $\sigma|_V$ denotes the \emph{restricted}
substitution of $\sigma$ w.r.t.\ $V$:
$\sigma|_V = \{ x \mapsto \applySubst{\sigma}{x} \mid x \in \Dom(\sigma) \cap V\}$.
A substitution $\sigma$ is called a \emph{renaming} if $\sigma$ is a bijection on $\cV$.
The \emph{composition} $\theta\Scomp\sigma$ {(simply $\theta\sigma$)
of substitutions $\sigma$ and $\theta$ is defined as $\ApplySubst{(\theta\Scomp\sigma)}{x}=\theta(\sigma(x))$.
A substitution $\sigma$ is called \emph{idempotent} if $\sigma\sigma = \sigma$ (i.e., $\Dom(\sigma)\cap\VRange(\sigma)=\emptyset$).
A substitution $\sigma$ is called \emph{more general than} a substitution $\theta$, written by $\sigma \leq \theta$, if there exists a substitution $\delta$ such that 
$\delta\sigma = \theta$. 
A finite set $E$ of term equations $\Eqn{s}{t}$ is called \emph{unifiable} if there exists a \emph{unifier} of $E$ such that $\applySubst{\sigma}{s} = \applySubst{\sigma}{t}$ for all term equations $\Eqn{s}{t}$ in $E$.
A \emph{most general unifier} (mgu) of $E$ is denoted by $\mgu(E)$ if $E$ is unifiable.
Terms $s$ and $t$ are called \emph{unifiable} if $\{\Eqn{s}{t}\}$ is unifiable.
The application of a substitution $\theta$ to $E$, denoted by $\applySubst{\theta}{E}$, is defined as $\applySubst{\theta}{E}=\{ \Eqn{\applySubst{\theta}{s}}{\applySubst{\theta}{t}} \mid \Eqn{s}{t} \in E\}$.

\subsection{Conditional Rewriting}

An \emph{oriented conditional rewrite rule} over a signature $\cF$ is
a triple $(\ell,r,c)$, denoted by $\CRule{\ell}{r}{c}$, such that the
\emph{left-hand side} $\ell$ is a non-variable term in $\Terms(\cF,\cV)$, the
\emph{right-hand side} $r$ is a term in $\Terms(\cF,\cV)$, and the
\emph{conditional part} $c$ is a sequence $\Condition{s_1}{t_1},\ldots,\Condition{s_k}{t_k}$ of term pairs ($k \geq 0$) where $s_1,t_1,\ldots,s_k,t_k \in \Terms(\cF,\cV)$. 
In particular, a conditional rewrite rule is called \emph{unconditional}
if the conditional part is the empty sequence (i.e., $k = 0$), and we
may abbreviate it to $\Rule{\ell}{r}$. 
Variables in $\Var(r,c)\setminus \Var(\ell)$ are called \emph{extra variables} of the rule.
An \emph{oriented conditional term rewriting system} (a CTRS, for short) over $\cF$ is a set of oriented conditional rewrite rules over $\cF$.
A CTRS is called an (unconditional) \emph{term rewriting system} (a TRS, for short)
if every rule $\CRule{\ell}{r}{c}$ in the CTRS is unconditional and
satisfies $\Var(\ell) \supseteq \Var(r)$.  
The \emph{reduction relation} $\to_\cR$ of a CTRS $\cR$ is defined as
${\to_{\cR}} = {\bigcup_{n \geq 0} \to_{(n),\cR}}$, where
${\to_{(0),\cR}} = \emptyset$, and ${\to_{(i+1),\cR}} = \{
(s[\applySubst{\sigma}{\ell}]_p,s[\applySubst{\sigma}{r}]_p) \mid s \in \Terms(\cF,\cV), ~ \CRule{\ell}{r}{\Condition{s_1}{t_1},\ldots,\Condition{s_k}{t_k}} \in \cR, ~ 
\applySubst{\sigma}{s_1} \mathrel{\to^*_{(i),\cR}} \applySubst{\sigma}{t_1}, ~\ldots,~ \applySubst{\sigma}{s_k} \mathrel{\to^*_{(i),\cR}} \applySubst{\sigma}{t_k} \}$ for $i \geq 0$.  
To specify the position where the rule is applied, we may write $\to_{p,\cR}$ instead of $\to_\cR$.
The \emph{underlying unconditional system} $\{ \Rule{\ell}{r} \mid \CRule{\ell}{r}{c} \in \cR \}$ of $\cR$ is denoted by $\cR_u$.
A term $t$ is called a \emph{normal form} (of $\cR$) if $t$ is
irreducible w.r.t.\ $\cR$. 
A substitution $\sigma$ is called \emph{normalized}
(w.r.t.\ $\cR$) if $\applySubst{\sigma}{x}$ is a normal form of $\cR$ for each variable $x \in \Dom(\sigma)$. 
A CTRS $\cR$ is called 
	\emph{Type 3} (3-CTRS, for short) if every rule $\CRule{\ell}{r}{c} \in \cR$ satisfies that $\Var(r) \subseteq \Var(\ell,c)$.
$\Var(s_i) \subseteq \Var(\ell,t_1,\ldots,t_{i-1})$ for all $1 \leq i \leq k$.

The sets of \emph{defined symbols} and \emph{constructors} of a CTRS $\cR$ over a signature $\cF$ are
denoted by $\cD_\cR$ and $\cC_\cR$, respectively: 
$\cD_\cR = \{ \Root(\ell) \mid \CRule{\ell}{r}{c} \in \cR \}$ and
$\cC_\cR = \cF \setminus \cD_\cR$. 
Terms in $\Terms(\cC_\cR,\cV)$ are called \emph{constructor terms of $\cR$}.
A substitution in $\Subst(\cC_\cR,\cV)$ is called a \emph{constructor substitution of $\cR$}.
A term of the form $\metasymb{f}(t_1,\ldots,t_n)$ with $\metasymb{f}/n \in \cD_\cR$ and $t_1,\ldots,t_n \in \Terms(\cC_\cR,\cV)$ is called \emph{basic}.
A CTRS $\cR$ is called a \emph{constructor system} if for every rule $\CRule{\ell}{r}{c}$ in $\cR$, $\ell$ is basic.
A 3-DCTRS $\cR$ is called \emph{syntactically deterministic} (an SDCTRS, for short) if for
every rule $\CRule{\ell}{r}{\Condition{s_1}{t_1},\ldots,\Condition{s_k}{t_k}}\in \cR$, every $t_i$ is a constructor term or a ground normal
form of $\cR_u$.

A CTRS $\cR$ is called \emph{operationally terminating} if there are no
infinite well-formed trees in a certain logical inference
system~\cite{LMM05}---operational termination means that the evaluation of conditions must either successfully terminate or fail in finite time.
Two terms $s$ and $t$ are said to be \emph{joinable}, written as $s \downarrow_\cR t$, if there exists a term $u$ such that $s \mathrel{\to^*_\cR} u \mathrel{\gets^*_\cR} t$.
A CTRS $\cR$ is called \emph{confluent} if $t_1\downarrow_\cR t_2$ for any terms $t_1,t_2$ such that $t_1 \mathrel{\gets^*_\cR} \cdot \mathrel{\to^*_\cR} t_2$.

\subsection{Innermost Conditional Narrowing}

We denote a pair of terms $s,t$ by $\Condition{s}{t}$ (not an equation $s \approx t$) because we analyze conditions of rewrite rules and distinguish the left- and right-hand sides of $\Condition{s}{t}$.
In addition, we deal with pairs of terms as terms by considering $\Condition{}{}$ a binary function symbol.
For this reason, we apply many notions for terms to pairs of terms without notice.
For readability, when we deal with $\Condition{s}{t}$ as a term, we often bracket it such as $(\Condition{s}{t})$.
As in~\cite{MH94}, any CTRS in this paper is assumed to implicitly include the rule $\Rule{\BraCondition{x}{x}}{\symb{\top}}$ where $\top$ is a special constant.
The rule $\Rule{\BraCondition{x}{x}}{\symb{\top}}$ is used to test structural equivalence between two terms $t_1,t_2$ by means of $\Condition{t_1}{t_2}$.

To deal with a conjunction of pairs $e_1, \ldots, e_k$ of terms ($e_i$ is either $\Condition{s_i}{t_i}$ or $\top$) as a term, we write $e_1 \Eand \Eqns \Eand e_k$ by using an associative binary symbol $\Eand$.
We call such a term an \emph{equational term}.
Unlike~\cite{NV13}, to avoid $\Eand$ to be a defined symbol, we do not use 
any rule for $\Eand$, e.g., $\Rule{(\top \Eand x)}{x}$.
Instead of derivations ending with $\top$, we consider derivations that end with terms in $\Tend$.
We assume that none of $\Eand$, $\Condition{}{}$, or $\top$ is included in the range of any substitution below.

In the following, for a constructor 3-SDCTRS $\cR$, a pair $\Condition{s}{t}$ of terms is called a \emph{goal} of $\cR$ if the left-hand side $s$ is either a constructor term or a basic term and the right-hand side $t$ is a constructor term.
An equational term is called a \emph{goal clause} of $\cR$ if it is a conjunction of goals for $\cR$.
Note that for a goal clause $T$, any instance $\applySubst{\theta}{T}$ with $\theta$ a constructor substitution is a goal clause.
\begin{example}
The equational term $\Condition{x < y}{\symb{true}} \Eand \Condition{y < x}{\symb{true}}$ is a goal clause of $\cRgcd$.
\end{example}

The \emph{narrowing} relation \cite{Sla74,Hul80}
mainly extends rewriting by replacing matching with unification.
This paper follows the formalization in~\cite{NV12lopstr}, while we use the rule $\Rule{\BraCondition{x}{x}}{\top}$ instead of the corresponding inference rule.
Let $\cR$ be a CTRS.
A goal clause 
$S=U \Eand \Condition{s}{t} \Eand S'$ with $U \in \Tend$ is said to  \emph{conditionally narrow} into an equational term $T$ \emph{at an innermost position}, written as $S \mathrel{\ileadsto_\cR} T$, if there exist a non-variable position $p$ of $\BraCondition{s}{t}$, a variant $\CRule{\ell}{r}{C}$ of a rule in $\cR$, and a constructor substitution $\sigma$ such that $\Var(\ell,r,C) \cap \Var(S)=\emptyset$, $(\Condition{s}{t})|_p$ is basic, $(\Condition{s}{t})|_p$ and $\ell$ are unifiable, $\sigma = \mgu(\{\Eqn{(\Condition{s}{t})|_p}{\ell}\})$, and 
$T= U \Eand \applySubst{\sigma}{C} \Eand \ApplySubst{\sigma}{(\Condition{s}{t})[r]_p} \Eand \applySubst{\sigma}{S'}$.
Note that all extra variables of $\CRule{\ell}{r}{C}$ remain in $T$ as \emph{fresh} variables which do not appear in $S$.
We assume that $\Var(S) \cap \VRange(\sigma|_{\Var((\Condition{s}{t})|_p)}) = \emptyset$ (i.e., $\sigma|_{\Var((\Condition{s}{t})|_p)}$ is idempotent) and 
$\Var((\Condition{s}{t})|_p) \subseteq \Dom(\sigma)$. 
We write $S \mathrel{\ileadsto_{\sigma|_{\Var(S)},\cR}} T$ to make the substitution explicit. 
%
An \emph{innermost narrowing derivation} $T_0 \mathrel{\ileadsto_{\sigma,\cR}^*} T_n$ (and $T_0 \mathrel{\ileadsto_{\sigma,\cR}^n} T_n$) denotes a
sequence of narrowing steps $T_0 \mathrel{\ileadsto_{\sigma_1,\cR}} \cdots \mathrel{\ileadsto_{\sigma_n,\cR}} T_n$ with $\sigma = (\sigma_n \cdots \sigma_1)|_{\Var(T_0)}$ 
an idempotent substitution.
When we consider two (or more) narrowing derivations $S_1 \mathrel{\ileadsto^*_{\sigma_1,\cR}} T_1$ and $S_2 \mathrel{\ileadsto^*_{\sigma_2,\cR}} T_2$, we assume that $\VRange(\sigma_1) \cap \VRange(\sigma_2) = \emptyset$.

Innermost narrowing is a counterpart of \emph{constructor-based rewriting} (cf.~\cite{NV12lopstr}).
Following~\cite{NV12lopstr}, we define \emph{constructor-based conditional rewriting} on goal clauses as follows: 
for a goal clause 
$S= U \Eand \Condition{s}{t} \Eand S'$ with $U \in \Tend$,
we write $S \mathrel{\cto_\cR} T$ if there exist a non-variable position $p$ of $\BraCondition{s}{t}$, a rule $\CRule{\ell}{r}{C}$ in $\cR$, and a constructor substitution $\sigma$ such that $(\Condition{s}{t})|_p$ is basic, $(\Condition{s}{t})|_p=\applySubst{\sigma}{\ell}$, and 
$T= U \Eand \applySubst{\sigma}{C} \Eand (\Condition{s}{t})[\applySubst{\sigma}{r}]_p \Eand S'$.
\begin{theorem}[\cite{NM18fscd}]
\label{thm:c-narrowing-soundness}
Let $\cR$ be a constructor SDCTRS, $T$ a goal clause, and $U \in \Tend$.
\begin{itemize}
	\item If\/ $T \mathrel{\ileadsto^*_{\sigma,\cR}} U$, then $\applySubst{\sigma}{T} \mathrel{\cto^*_\cR} U$ (i.e., 
		$\applySubst{\sigma}{s} \mathrel{\cto^*_\cR} \applySubst{\sigma}{t}$ for all goals $\Condition{s}{t}$ in $T$).
	\item For a constructor substitution $\theta$, if\/ $\applySubst{\theta}{T} \mathrel{\cto^*_\cR} U$, then there exists an idempotent constructor substitution $\sigma$ such that $T \mathrel{\ileadsto^*_{\sigma,\cR}} U$ and $\sigma \leq \theta$.
\end{itemize}
\end{theorem}

\begin{example}
\label{ex:conditional-narrowing}
Consider $\cRgcd$ in Section~\ref{sec:intro} again.
The following is an instance of innermost conditional narrowing of $\cRgcd$:
\[
\begin{array}{@{}l@{~~~~}l@{}}
\lefteqn{%
\BraCondition{\symb{gcd}(\symb{s}^4(\symb{0}),y)}{z} \Eand \BraCondition{\symb{s}(\symb{0}) < z}{\symb{true}}
} 
\\
& \mathrel{\ileadsto_{\{y\mapsto \symb{s}(y_1)\},\cRgcd}}
\BraCondition{y_1 < \symb{s}^3(\symb{0})}{\symb{true}} \Eand \BraCondition{\symb{gcd}(\symb{s}^3(\symb{0}) - y_1,\symb{s}(y_1))}{z} \Eand \BraCondition{\symb{s}(\symb{0}) < z}{\symb{true}}
\\	
& \mathrel{\ileadsto^2_{\{y_1\mapsto \symb{s}(\symb{0})\},\cRgcd}}
\BraCondition{\symb{true}}{\symb{true}} \Eand \BraCondition{\symb{gcd}(\symb{s}^3(\symb{0}) - \symb{s}(\symb{0}),\symb{s}^2(\symb{0}))}{z} \Eand \BraCondition{\symb{s}(\symb{0}) < z}{\symb{true}}
\\	
& \mathrel{\ileadsto_{\ID,\cRgcd}}
\top \Eand \BraCondition{\symb{gcd}(\symb{s}^3(\symb{0}) - \symb{s}(\symb{0}),\symb{s}^2(\symb{0}))}{z} \Eand \BraCondition{\symb{s}(\symb{0}) < z}{\symb{true}}
\\	
& \mathrel{\ileadsto^2_{\ID,\cRgcd}}
\top \Eand \BraCondition{\symb{gcd}(\symb{s}^2(\symb{0}),\symb{s}^2(\symb{0}))}{z} \Eand \BraCondition{\symb{s}(\symb{0}) < z}{\symb{true}}
\\	
& \mathrel{\ileadsto_{\ID,\cRgcd}}
\top \Eand \BraCondition{\symb{s}^2(\symb{0})}{z} \Eand \BraCondition{\symb{s}(\symb{0}) < z}{\symb{true}}
\\	
& \mathrel{\ileadsto_{\{z\mapsto \symb{s}^2(\symb{0})\},\cRgcd}}
\top \Eand \top \Eand \BraCondition{\symb{s}(\symb{0}) < \symb{s}^2(\symb{0})}{\symb{true}}
\\	
& \mathrel{\ileadsto^2_{\ID,\cRgcd}}
\top \Eand \top \Eand \BraCondition{\symb{true}}{\symb{true}}
\\	
& \mathrel{\ileadsto_{\ID,\cRgcd}}
\top \Eand \top \Eand \top 
\\	
\end{array}
\]
The following constructor-based rewriting derivation corresponds to the above narrowing derivation:
\[
\begin{array}{@{}l@{~~~~}l@{}}
\lefteqn{%
\BraCondition{\symb{gcd}(\symb{s}^4(\symb{0}),\symb{s}^2(\symb{0}))}{\symb{s}^2(\symb{0})} \Eand \BraCondition{\symb{s}(\symb{0}) < \symb{s}^2(\symb{0})}{\symb{true}}
} 
\\
& \mathrel{\cto_{\cRgcd}}
\BraCondition{\symb{s}(\symb{0}) < \symb{s}^3(\symb{0})}{\symb{true}} \Eand \BraCondition{\symb{gcd}(\symb{s}^3(\symb{0}) - \symb{s}(\symb{0}),\symb{s}^2(\symb{0}))}{\symb{s}^2(\symb{0})} \Eand \BraCondition{\symb{s}(\symb{0}) < \symb{s}^2(\symb{0})}{\symb{true}}
\\	
& \mathrel{\cto_{\cRgcd}}
\BraCondition{\symb{true}}{\symb{true}} \Eand \BraCondition{\symb{gcd}(\symb{s}^3(\symb{0}) - \symb{s}(\symb{0}),\symb{s}^2(\symb{0}))}{\symb{s}^2(\symb{0})} \Eand \BraCondition{\symb{s}(\symb{0}) < \symb{s}^2(\symb{0})}{\symb{true}}
\\	
& \mathrel{\cto_{\cRgcd}}
\top \Eand \BraCondition{\symb{gcd}(\symb{s}^3(\symb{0}) - \symb{s}(\symb{0}),\symb{s}^2(\symb{0}))}{\symb{s}^2(\symb{0})} \Eand \BraCondition{\symb{s}(\symb{0}) < \symb{s}^2(\symb{0})}{\symb{true}}
\\	
& \mathrel{\cto^2_{\cRgcd}}
\top \Eand \BraCondition{\symb{gcd}(\symb{s}^2(\symb{0}),\symb{s}^2(\symb{0}))}{\symb{s}^2(\symb{0})} \Eand \BraCondition{\symb{s}(\symb{0}) < \symb{s}^2(\symb{0})}{\symb{true}}
\\	
& \mathrel{\cto_{\cRgcd}}
\top \Eand \BraCondition{\symb{s}^2(\symb{0})}{\symb{s}^2(\symb{0})} \Eand \BraCondition{\symb{s}(\symb{0}) < \symb{s}^2(\symb{0})}{\symb{true}}
\\	
& \mathrel{\cto_{\cRgcd}}
\top \Eand \top \Eand \BraCondition{\symb{s}(\symb{0}) < \symb{s}^2(\symb{0})}{\symb{true}}
\\	
& \mathrel{\cto^2_{\cRgcd}}
\top \Eand \top \Eand \BraCondition{\symb{true}}{\symb{true}}
\\	
& \mathrel{\cto_{\cRgcd}}
\top \Eand \top \Eand \top 
\\	
\end{array}
\]
\end{example}

\subsection{Regular Tree Grammars}

A \emph{regular tree grammar} (an RTG, for short) is a quadruple $\cG=(S,\cN,\cF,\cP)$ such that $\cF$ is a signature, $\cN$ is a finite set of \emph{non-terminals} (constants not in $\cF$), $S \in \cN$, and $\cP$ is a finite set of \emph{production rules} of the form $\Rule{A}{\beta}$ with $A \in \cN$ and $\beta \in \Terms(\cF\cup\cN)$.
Given a non-terminal $S' \in \cN$, the set $\{ t \in \Terms(\cF) \mid S' \mathrel{\to_\cP^*} t \}$ is the \emph{language generated by $\cG$ from $S'$}, denoted by $\Lang(\cG,S')$.
The \emph{initial} non-terminal $S$ is not so relevant in this paper.
A \emph{regular tree language} is a language generated by an RTG from one of its non-terminals.
The class of regular tree languages is equivalent to the class of \emph{recognizable tree languages} which are recognized by \emph{tree automata}.
This means that the \emph{intersection (non-)emptiness problem} for regular tree languages is decidable.

\begin{example}
The RTG $\cGex=(X,\{X,X'\},\{\symb{0}/0,\symb{s}/1\},\{\Rule{X}{\symb{0}}, ~ \Rule{X}{\symb{s}(X')}, ~ \Rule{X'}{\symb{s}(X)} \})$ generates the sets of even and odd numbers over $\symb{0}$ and $\symb{s}$ from $X$ and $X'$, respectively:
$\Lang(\cGex,X) =\Lang(\cGex) = \{ \symb{s}^{2n}(\symb{0}) \mid n \geq 0 \}$ and $\Lang(\cGex,X') = \{ \symb{s}^{2n+1}(\symb{0}) \mid n \geq 0 \}$.
\end{example}

\section{Coding of Tuples of Ground Terms}
\label{sec:coding}

In this section, we introduce the notion of \emph{coding} of tuples of ground terms~\cite[Section~3.2.1]{TATA}.
To simplify discussions, we consider pairs of terms.

Let $\cF$ be a signature. 
We prepare the signature $\cF' = (\cF\cup\{\bot\})^2$, where $\bot$ is a new constant.
For symbols $\symb{f_1},\symb{f_2}\in \cF$, we denote the function symbol $(\symb{f_1},\symb{f_2}) \in \cF'$ by $\Pairsymb{f_1}{f_2}$, and the arity of  $\Pairsymb{f_1}{f_2}$ is $\max(\Arity(\symb{f_1}),\Arity(\symb{f_2}))$.
The coding of pairs of ground terms, $\Coding{\cdot}{\cdot}$, is recursively defined as follows:
\begin{itemize}
	\item 
$
\Coding{\symb{f}(s_1,\ldots,s_m)}{\symb{g}(t_1,\ldots,t_n)} =
\Pairsymb{f}{g}(\Coding{s_1}{t_1},\ldots,\Coding{s_m}{t_m},\Coding{\bot}{t_{m+1}},\ldots,\Coding{\bot}{t_n})
$
if $m \leq n$, 
\item 
$
\Coding{\symb{f}(s_1,\ldots,s_m)}{\symb{g}(t_1,\ldots,t_n)} =
\Pairsymb{f}{g}(\Coding{s_1}{t_1},\ldots,\Coding{s_n}{t_n},\Coding{s_{n+1}}{\bot},\ldots,\Coding{s_m}{\bot})
$ if $m>n$,
	\item 
$
\Coding{\symb{f}(s_1,\ldots,s_m)}{\bot} =
\Pairsymb{f}{\bot}(\Coding{s_1}{\bot},\ldots,\Coding{s_m}{\bot})
$,
and
\item 
$
\Coding{\bot}{\symb{g}(t_1,\ldots,t_n)} =
\Pairsymb{\bot}{g}(\Coding{\bot}{t_1},\ldots,\Coding{\bot}{t_n})
$.
\end{itemize}
Note that $\Pos(\Coding{t_1}{t_2}) = \Pos(t_1) \cup \Pos(t_2)$.
Note also that for $i=1,2$ and for $p\in \Pos(\Coding{t_1}{t_2})$, if $p \notin \Pos(t_i)$, then $\bot$ is complemented for $t_i$.
As described in~\cite[Section~3.2.1]{TATA}, the basic idea of coding is to stack function symbols as illustrated in Figure~\ref{fig:ideaOfCoding}.

\begin{example}
As in Figure~\ref{fig:ideaOfCoding}, 
$
\Coding{\symb{f}(\symb{g}(\symb{a}),\symb{g}(\symb{a}))}{\symb{f}(\symb{f}(\symb{a},\symb{a}),\symb{a})}
=
\Pairsymb{f}{f}(\Pairsymb{g}{f}(\Pairsymb{a}{a}, \Pairsymb{\bot}{a}), \Pairsymb{g}{a}(\Pairsymb{a}{\bot}))
$.	
\end{example}

\section{Grammar Representations for Sets of Idempotent Substitutions}
\label{sec:grammar_representation}

In this section, we briefly introduce grammar representations that define sets of idempotent substitutions.
We follow the formalization in~\cite{NM18fscd}, which is based on \emph{success set equations} in~\cite{NV13}.
Since substitutions derived by narrowing steps are assumed to be idempotent, we deal with only idempotent substitutions which introduce only \emph{fresh} variables not appearing in any previous term.


In the following, a renaming $\xi$ is used to (partially) rename a particular term $t$ w.r.t.\ a set $X$ of variables with $X \subseteq \Var(t)\cap\Dom(\xi)$ by assuming that $\xi|_X$ is injective on $X$ (i.e., for all variables $x,y \in X$, if $x\ne y$ then $\applySubst{\xi}{x} \ne \applySubst{\xi}{y}$) and $\VRange(\xi|_X) \cap (\Var(t)\setminus X) = \emptyset$.
For this reason, 
we write $\xi|_X$ instead of $\xi$, and call $\xi|_X$ a \emph{renaming for $t$} (simply, a renaming).

We first introduce terms to represent idempotent substitutions computed using composition operators $\Scomp$ and $\Pcomp$.
We prepare the signature $\Sigma$ consisting of the following symbols~\cite{NM18fscd}:
\begin{itemize}
	\item 
	a finite number of idempotent substitutions which are considered constants,
	\hfill (basic elements)
	\item 
	a constant $\EmptysetSym$,
	\hfill (the empty set/non-existence)
	\item 
	an associative binary symbol ${\ScompSym}$,
	\hfill(standard composition)
	\item 
	an associative binary symbol ${\Eand}$, and
	\hfill (parallel composition)
   \item 
	a binary symbol $\CALL$.
	\hfill  (recursion with renaming)
\end{itemize}
We use infix notation for $\ScompSym$ and $\Eand[V]$, and may omit brackets with the precedence such that $\ScompSym$ has a higher priority than $\Eand$.

We deal with terms over $\Sigma$ and some constants used for non-terminals of grammar representations, where we allow such constants to only appear in the first argument of $\CALL$.
Note that a term without any constant may appear in the first argument of $\CALL$.
Given a finite set $\cN$ of constants ($\Sigma \cap \cN = \emptyset$), we denote the set of such terms by $\Terms(\Sigma\cup\cN)$.
We assume that each constant in $\cN$ has a term $t$ (possibly a goal clause) as subscript such as $\Gamma_t$.
For an expression $\CALL(\Gamma_t,\delta)$, the role of $\Gamma_t$ is to generate substitutions (more precisely, terms in $\Terms(\Sigma)$) from $\Gamma_t$, e.g., recursively, and the role of $\delta$ is to connect such substitutions with other substitutions if necessary, where the application of $\delta$ to some term results in $t$.
For this reason, we restrict the second argument of $\CALL$ to renamings, and for each term $\CALL(\Gamma_t,\delta)$, we require $\delta$ to be an idempotent renaming (i.e., $\Dom(\delta)\cap\VRange(\delta)=\emptyset$ and $\delta$ is injective on $\Dom(\delta)$) such that $\VRange(\delta) \subseteq \Var(t)$, and $\Dom(\delta) \cap (\Var(t)\setminus \VRange(\delta)) = \emptyset$.
\begin{example}[\cite{NM18fscd}]
\label{ex:expressions}
The following are terms in $\Terms(\Sigma)$: 
\begin{itemize}
	\item 
	$\{ y \mapsto \symb{0} \} \ScompSym \{x \mapsto \symb{s}(y)\}$,
	\item 
	$(\{x'\mapsto \symb{s}(y) \} \ScompSym \{x\mapsto x'\}) \Eand[\{x\}] \{x\mapsto \symb{s}(\symb{s}(z))\}$,
	\item 
	$(
			\EmptysetSym
			\Eand[\{y\}]
			\{ y \mapsto z \}
			)
			\ScompSym 
			\{x \mapsto \symb{s}(y)\}$, and
	\item 
	$\CALL(\{x\mapsto \symb{0}, ~ y \mapsto \symb{s}(y')\},\{x'\mapsto x, ~ y'\mapsto y\}) \ScompSym \{ y\mapsto\symb{s}(x')\}$.
\end{itemize}
Note that substitutions $\{ y \mapsto \symb{0} \}$, $\{x \mapsto \symb{s}(y)\}$, $\{x'\mapsto \symb{s}(y) \}$, $\{x\mapsto x'\}$, $\{x\mapsto \symb{s}(\symb{s}(z))\}$, $\{ y \mapsto z \}$, $\{x\mapsto \symb{0}, ~ y \mapsto \symb{s}(y')\}$, $\{x'\mapsto x, ~ y'\mapsto y\}$, $\{ y\mapsto\symb{s}(x')\}$ are considered constants.
\end{example}

Next, we recall \emph{parallel composition} $\Pcomp$ of idempotent substitutions~\cite{HR89,Pal90}, which is one of the most important key operations to enable us to construct \emph{finite} narrowing trees.
Given a substitution $\theta=\{x_1\mapsto t_1, ~ \ldots, ~ x_n \mapsto t_n\}$, we denote the set of term equations $\{\Eqn{x_1}{t_1},\, \ldots,\,  \Eqn{x_n}{t_n}\}$ by $\widehat{\theta}$.
\begin{definition}[parallel composition $\Pcomp$~\cite{Pal90}]
\label{def:pcomp}
Let $\theta_1$ and $\theta_2$ be idempotent substitutions.
Then, we define $\Pcomp$ as follows:
$\theta_1 \Pcomp \theta_2 = \mgu( \widehat{\theta_1} \cup \widehat{\theta_2} )$ if $\widehat{\theta_1} \Eand{} \widehat{\theta_2}$ is unifiable, and otherwise, $\theta_1 \Pcomp \theta_2 = \Fail$.
Note that we define $\theta_1 \Pcomp \theta_2 = \Fail$ if $\theta_1$ or $\theta_2$ is not idempotent.
Parallel composition is extended to sets $\Theta_1, \Theta_2$ of idempotent substitutions in the natural way:
$
	\Theta_1 \Pcomp \Theta_2 = \{ \theta_1 \Pcomp \theta_2 \mid \theta_1 \in \Theta_1, ~ \theta_2 \in \Theta_2, ~ \theta_1 \Pcomp \theta_2 \ne \Fail \}
$.
\end{definition}
We often have two or more substitutions that can be results of $\theta_1 \Pcomp \theta_2$ ($\ne \Fail$), while most general unifiers are unique up to variable renaming.
To simplify the semantics of grammar representations for substitutions, as a result of $\theta_1 \Pcomp \theta_2$ ($\ne \Fail$), we adopt an idempotent substitution $\sigma$ such that $\Dom(\theta_1)\cup\Dom(\theta_2) \subseteq \Dom(\sigma)$.
Note that most general unifiers we can adopt as results of $\theta_1 \Pcomp \theta_2$ under the convention are still not unique, while they are unique up to variable renaming.

\begin{example}[\cite{NM18fscd}]
	The parallel composition $\{x \mapsto \symb{s}(z), ~ y \mapsto z\} \Pcomp \{ x \mapsto w\}$ may return $\{x \mapsto \symb{s}(z), ~ y \mapsto z, ~ w \mapsto \symb{s}(z) \}$, but we do not allow $\{x \mapsto \symb{s}(y), ~ z \mapsto y, ~ w \mapsto \symb{s}(y) \}$ as a result because $y$ appears in the range.
	On the other hand, $\{x \mapsto \symb{s}(z), ~ y \mapsto z\} \Pcomp \{ x \mapsto y\} = \Fail$.
\end{example}

A key of construction of narrowing trees (and their grammar representations) is \emph{compositionality} of innermost narrowing (cf.~\cite{NM18fscd}):
$S_1 \Eand S_2 \mathrel{\ileadsto^*_{\sigma,\cR}} T$ if and only if $S_1 \mathrel{\ileadsto^*_{\sigma_1,\cR}} T_1$, $S_2 \mathrel{\ileadsto^*_{\sigma_2,\cR}} T_2$, $T=T_1 \Eand T_2$, and $\sigma = \sigma_1 \Pcomp \sigma_2$.
To compute a substitution derived by innermost narrowing from a goal clause $S_1 \Eand S_2$, we compute substitutions $\sigma_1$ and $\sigma_2$ derived by innermost narrowing from $S_1$ and $S_2$, respectively, and then compute $\sigma_1 \Pcomp \sigma_2$.
When we compute $\sigma_1 \Pcomp \sigma_2$
from two narrowing derivations $S_1 \mathrel{\ileadsto^*_{\sigma_1,\cR}} T_1$ and $S_2 \mathrel{\ileadsto^*_{\sigma_2,\cR}} T_2$, 
 we assume that $\VRange(\sigma_1) \cap \VRange(\sigma_2) = \emptyset$.
To satisfy this assumption explicitly in the semantics for $\Terms(\Sigma)$, we introduce an operation $\RenameFresh{\cdot}{\delta}$ of substitutions to make a substitution introduce only variables that do not appear in $\Dom(\delta)\cup\VRange(\delta)$:
for substitutions $\sigma,\delta$, we define $\RenameFresh{\sigma}{\delta}$ by $(\xi \Scomp \sigma)|_{\Dom(\sigma)}$ where $\xi$ is a renaming such that $\Dom(\xi) = \VRange(\sigma)$ and $\VRange(\xi)\cap (\Dom(\delta)\cup\VRange(\delta)\cup\Dom(\sigma))=\emptyset$.%
\footnote{
For $\VRange(\xi)$, we choose variables not appearing in any substitutions in $\Sigma$.
}
The subscript $\delta$ of $\RenameFresh{\cdot}{\delta}$ is used to specify freshness of variables---we say that a variable $x$ is \emph{fresh} w.r.t.\ a set $X$ of variables if $x \notin X$.

A term $e$ in $\Terms(\Sigma)$ defines a substitution.
The semantics of terms in $\Terms(\Sigma)$ is inductively defined as follows~\cite{NM18fscd}:
\begin{itemize}
	\item 
$\Semantics[Subst]{\theta}{X} = \theta$ if $\theta$ is a substitution,
	\item 
$\Semantics[Subst]{e_1 \ScompSym e_2}{X} = \Semantics[Subst]{e_1}{Y} \Scomp \Semantics[Subst]{e_2}{X}$ if $\Semantics[Subst]{e_2}{X}\ne \Fail$ and $\Semantics[Subst]{e_1}{Y}\ne \Fail$,
	\item 
$\Semantics[Subst]{e_1 \Eand[V] e_2}{X} = (\theta_1 \Pcomp \theta_2)|_{{\Dom(\theta_1)}\cup{\Dom(\theta_2)}}$ if $\Semantics[Subst]{e_1}{X} \ne \Fail$ and $\Semantics[Subst]{e_2}{Y}\ne \Fail$, where $\theta_1=\Semantics[Subst]{e_1}{X}$ and $\theta_2=\RenameFresh{\Semantics[Subst]{e_2}{Y}}{\theta_1}$,
	\item 
	$\Semantics[Subst]{\CALL(e,\delta)}{X} = (\RenameFresh{\Semantics[Subst]{e}{Y}}{\delta} \Scomp \delta)|_{\Dom(\delta)}$ 
		if $\Semantics[Subst]{e}{X}\ne \Fail$ and $\VRange(\delta) \subseteq \Dom(\Semantics[Subst]{e}{X})$, and
	\item 
otherwise, $\Semantics[S]{e}{X} = \Fail$ 	(e.g., $\Semantics[Subst]{\EmptysetSym}{X}=\Fail$).
\end{itemize}
Notice that $\Gamma_t$, a non-terminal used in an RTG, is not included in $\Terms(\Sigma)$, and thus, $\Semantics[Subst]{\Gamma_t}{X}$ is not defined.
Since $\Pcomp$ may fail, we allow to have $\Fail$, e.g.,
$\Semantics[Subst]{\{ y \mapsto \symb{s}(x)\} \ScompSym \{ x \mapsto y \} \Eand \{z \mapsto \symb{0}\}}{X} = \Fail$.
The number of variables appearing in an RTG defined below is finite.
However, we would like to use RTGs to define infinitely many substitutions such that the maximum number of variables we need cannot be fixed.
To solve this problem, in the definition of $\Semantics[Subst]{\CALL(e,\delta)}{X}$, we introduced the operation $\RenameFresh{\cdot}{\delta}$ that makes all variables introduced by $\Semantics[Subst]{e}{X\cup\Dom(\delta)\cup\VRange(\delta)}$ \emph{fresh w.r.t.\ $\Dom(\delta)\cup\VRange(\delta)$}.
In~\cite{NV13}, this operation is implicitly considered, but in~\cite{NM18fscd}, $\CALL$ is explicitly introduced to the syntax in order to convert terms in $\Terms(\Sigma)$ precisely.
To assume $\VRange(\Semantics[Subst]{e_1}{X}) \cap \VRange(\Semantics[Subst]{e_2}{X\cup\Dom(\theta_1)\cup\VRange(\theta_1)}) = \emptyset$ for $\Semantics[Subst]{e_1 \Eand[V] e_2}{X}$, we also introduced $\RenameFresh{\cdot}{\theta_1}$ in the case of $\Semantics[Subst]{e_1 \Eand[V] e_2}{X}$.

The semantics of terms in $\Terms(\Sigma)$ is naturally extended to subsets of $\Terms(\Sigma)$ as follows:
for a set $L \subseteq \Terms(\Sigma)$, 
$\Semantics{L}{} = \{ \Semantics[Subst]{e}{\Var(t)} \mid e \in L, ~ \Semantics[Subst]{e}{\Var(t)} \ne \Fail \}
$.

\begin{example}[\cite{NM18fscd}]
The expressions in Example~\ref{ex:expressions} are interpreted as follows:
\begin{itemize}
	\item 
	$\Semantics[Subst]{\{ y \mapsto \symb{0} \} \ScompSym \{x \mapsto \symb{s}(y)\}}{X} 
	= \{ y \mapsto \symb{0} \} \Scomp \{x \mapsto \symb{s}(y)\} 
	= \{x \mapsto \symb{s}(\symb{0}), ~ y \mapsto \symb{0} \}$,
	\item 
	$\Semantics[Subst]{(\{x'\mapsto \symb{s}(y) \} \ScompSym \{x\mapsto x'\}) \Eand[\{x\}] \{x\mapsto \symb{s}(\symb{s}(z))\}}{\{x\}}$\\
			$= \left(\{ x\mapsto \symb{s}(y), ~ x'\mapsto \symb{s}(y) \} \Pcomp \RenameFresh{\{x\mapsto \symb{s}(\symb{s}(z'))\}}{\{x\mapsto \symb{s}(y), ~ x'\mapsto \symb{s}(y)\}} \right)|_{\{x,x'\}}$\\
			$= \left(\{ x\mapsto \symb{s}(y), ~ x'\mapsto \symb{s}(y) \} \Pcomp \{x\mapsto \symb{s}(\symb{s}(z'))\} \right)|_{\{x,x'\}}$\\
			$= \left(\{ x\mapsto \symb{s}(\symb{s}(z')), ~ x'\mapsto \symb{s}(\symb{s}(z')) \} \right)|_{\{x,x'\}}$\,%
			\footnote{ Note that $\{ x\mapsto \symb{s}(y), ~ x'\mapsto \symb{s}(y) \} \Pcomp \{x\mapsto \symb{s}(\symb{s}(z'))\}=\{ x\mapsto \symb{s}(\symb{s}(z')), ~ x'\mapsto \symb{s}(\symb{s}(z')), ~ y \mapsto \symb{s}(z') \}$.}\\
			$= \{x\mapsto \symb{s}(\symb{s}(z')), ~ x'\mapsto \symb{s}(\symb{s}(z')) \}$,
	\item 
	$\Semantics[Subst]{
				(
				\EmptysetSym
				\Eand[\{y\}]
				\{ y \mapsto z \}
				)
				\ScompSym 
				\{x \mapsto \symb{s}(y)\} }{} = \Fail$ (since $\Semantics{\EmptysetSym}{}=\Fail$ and then $\Semantics[Subst]{\EmptysetSym
				\Eand[\{y\}]
				\{ y \mapsto z \}
				}{} = \Fail$), and
	\item 
	$\Semantics[Subst]{\CALL(\{x\mapsto \symb{0}, ~ y \mapsto \symb{s}(y')\},\{x'\mapsto x, ~ y'\mapsto y\}) \ScompSym \{ y\mapsto\symb{s}(x')\}}{\{y\}}$\\
			$= \left( \RenameFresh{\{x\mapsto \symb{0}, ~ y \mapsto \symb{s}(y')\}}{\{x'\mapsto x, ~ y'\mapsto y\}} \Scomp \{x'\mapsto x, ~ y'\mapsto y\}\right)|_{\{x',y'\}} \Scomp \{ y\mapsto\symb{s}(x')\}$\\
			$= \left( \{x\mapsto \symb{0}, ~ y \mapsto \symb{s}(y'')\} \Scomp \{x'\mapsto x, ~ y'\mapsto y\} \right)|_{\{x',y'\}} \Scomp \{ y\mapsto\symb{s}(x')\}$\\
			$= \{ x' \mapsto \symb{0}, ~ y' \mapsto \symb{s}(y''), ~ y \mapsto \symb{s(\symb{0})} \}$.
\end{itemize}
\end{example}

To define sets of idempotent substitutions, we adopt RTGs.
In the following, we drop the third component from grammars constructed below because the third one is fixed to $\Sigma$ with a finite number of substitutions that are clear from production rules.
A \emph{substitution-set grammar} (SSG) for a term $t_0$ is an RTG $\cG=(\Gamma_{t_0},\cN,\cP)$ such that $\cN$ is a finite set of non-terminals $\Gamma_t$, $\Gamma_{t_0} \in \cN$, and $\cP$ is a finite set of production rules of the form $\SSGRule{\Gamma_t}{\beta}$ with $\beta\in \Terms({\Sigma} \cup {\cN})$.
Note that $L(\cG,\Gamma_t) = \{ e \in \Terms(\Sigma) \mid \Gamma_t \mathrel{\to_{\cG}^*} e \}$ for each $\Gamma_t \in \cN$, and the numbers of variables appearing in $L(\cG,\Gamma_t)$ is finite.
The set of substitutions generated by $\cG$ from $\Gamma_t \in \cN$ is $\SubstSet{\cG}{\Gamma_t}$, i.e., 
$
\SubstSet{\cG}{\Gamma_t} = \{ \Semantics[Subst]{e}{\Var(t)} \mid e \in \Lang(\cG,\Gamma_t), ~ \Semantics[Subst]{e}{\Var(t)} \ne \Fail \}
$.
Note that the number of variables in $\bigcup_{\theta \in \SubstSet{\cG}{\Gamma_t}} \VRange(\theta)$ may be infinite because of the interpretation for $\CALL$.

\begin{example}
The RTG $\cGgcd$ in Section~\ref{sec:intro} is an SSG for a term $\Gamma_{\Condition{x < y}{\symb{true}} \Eand \Condition{y < x}{\symb{true}}}$.
We have that
\[
\begin{array}{@{}l@{}}
L(\cGgcd,\Gamma_{\Condition{x < y}{\symb{true}}}) = \\
~
\left\{
\begin{array}{@{\,}l@{\,}}
 \{x \mapsto \symb{0}, ~ y \mapsto \symb{s}(y_2) \},\\[3pt]
 \CALL(\{x \mapsto \symb{0}, ~ y \mapsto \symb{s}(y_2) \},\{x_3\mapsto x, ~ y_3 \mapsto y\}) \ScompSym \{x\mapsto \symb{s}(x_3), ~ y\mapsto \symb{s}(y_3)\}, 
 \\[3pt]
 \CALL\left(
   \begin{array}{@{}c@{}}
   \CALL(\{x \mapsto \symb{0}, ~ y \mapsto \symb{s}(y_2) \},\{x_3\mapsto x, ~ y_3 \mapsto y\}) \\
     \ScompSym \\
   \{x\mapsto \symb{s}(x_3), ~ y\mapsto \symb{s}(y_3)\} \\
   \end{array}, 
 \{x_3\mapsto x, ~ y_3 \mapsto y\}  
 \right) \ScompSym \{x\mapsto \symb{s}(x_3), ~ y\mapsto \symb{s}(y_3)\},
 \\ ~~ \ldots
 \end{array}
 \right\}
 \\
 \end{array}
\] 
and 
$\SubstSet{\cGgcd}{\Gamma_{\Condition{x < y}{\symb{true}}}} = \{ \{ x \mapsto \symb{s}^m(\symb{0}), ~ y \mapsto \symb{s}^n(a) \} \mid 0 \leq m < n, ~ a \in \{\symb{0},\symb{true},\symb{false} \} \}$.
\end{example}


\section{Transforming SSGs into RTGs Generating Ranges of Substitutions}
\label{sec:transformation}

In this section, given a goal clause $T$ and two variables $x_1,x_2$ appearing in $T$, we show a transformation of an SSG $\cG=(\Gamma_{T_0},\cN,\cP)$ into an RTG $\cG'$ such that $\Lang(\cG',\Gamma_T^{(x_1,x_2)}) \supseteq \{ \Coding{\applySubst{\xi\theta}{x_1}}{\applySubst{\xi\theta}{x_2}} \mid \theta \in \Lang(\cG,\Gamma_T), ~ \xi \in \Subst(\cC), ~ \Var(\applySubst{\theta}{x_1},\applySubst{\theta}{x_2}) \subseteq \Dom(\xi) \}$, where $\cC$ is a set of constructors we deal with.
Note that $T$ does not have to be $T_0$.
The transformation is an extension of the transformation in~\cite[Section~7]{NM18fscd} and applicable to SSGs satisfying a certain syntactic condition shown later.
In the following, 
we aim at showing that $\Lang(\cGgcd,\Gamma_{\Condition{x < y}{\symb{true}}})\cap \Lang(\cGgcd,\Gamma_{\Condition{y < x}{\symb{true}}})=\emptyset$.
We use $\cC$ as a set of constructors unless noted otherwise.

Let $\cG$ be an SSG $(\Gamma_{T_0},\cN,\cP)$ and $T$ a goal clause such that $\Gamma_T \in \cN$.
We denote by $\cP|_{\Gamma_T}$ the set of production rules that are reachable from $\Gamma_T$.
We assume that any rule in $\cP|_{\Gamma_T}$ is of the following form:
\[
	\Rule{\Gamma_{T'}}{\theta_1 \mid \cdots \mid \theta_m \mid \CALL(\Gamma_{T_1},\delta_1) \ScompSym \theta_{m+1} \mid \cdots \mid \CALL(\Gamma_{T_n},\delta_n) \ScompSym \theta_{m+n}}
\]
where $\VRange(\delta_j) = \Var(T_j)$%
\footnote{\label{fnt:response}
 In defining SSGs, we only required that $\VRange(\delta_j) \subseteq \Var(T_j)$, but to make the transformation below precise, we require that $\VRange(\delta_j) = \Var(T_j)$.
This requirement is not restrictive because SSGs for narrowing trees satisfy this requirement because $\delta_j$ connects $T_j$ with a renamed variant which has no shared variable with $T_j$.}
for all $1 \leq j \leq n$, and  $\theta_1,\ldots,\theta_{m+n}$ are idempotent substitutions such that $\Dom(\theta_j) = \Var(T')$ for all $1 \leq j \leq m+n$.
Note that $\Rule{\Gamma_{T'}}{\CALL(\Gamma_{T''},\delta)}$ is considered $\Rule{\Gamma_{T'}}{\CALL(\Gamma_{T''},\delta) \ScompSym \ID}$.
In addition, for each $\Rule{\Gamma_{T'}}{\CALL(\Gamma_{T_i},\delta_i) \ScompSym \theta_{m+i}}$ with $1 \leq i \leq n$, 
we assume that 
for all variables $x,y$ in $T'$ and for each position $p \in 
\Pos(\applySubst{\delta\theta_{m+i}}{x})\cap\Pos(\applySubst{\delta\theta_{m+i}}{y})$,
	all of the following hold:
	\begin{itemize}
		\item if $(\applySubst{\delta\theta_{m+i}}{x})|_p \in \Var(T_i)$, then $(\applySubst{\delta\theta_{m+i}}{y})|_p \in \Var(T_i) \cup \Terms(\cC,\cV\setminus\Var(T_i))$,
		and
		\item if $(\applySubst{\delta\theta_{m+i}}{y})|_p \in \Var(T_i)$, then $(\applySubst{\delta\theta_{m+i}}{x})|_p \in \Var(T_i) \cup \Terms(\cC,\cV\setminus\Var(T_i))$.
	\end{itemize}
This assumption implies that 
for such $x$, $y$, and $p$,
the terms $(\applySubst{\delta\theta_{m+i}}{x})|_p$ and $(\applySubst{\delta\theta_{m+i}}{y})|_p$ satisfy one of the following:
	\begin{enumerate}
		\renewcommand{\labelenumi}{(\alph{enumi})}
		\item both are rooted by function symbols,
		\item both are variables in $\Var(T_i)$,
		\item one is a variable in $\Var(T_i)$ and the other is a term in $\Terms(\cC,\cV\setminus\Var(T_i))$,
			or
		\item both are terms in $\Terms(\cC,\cV\setminus\Var(T_i))$.
	\end{enumerate}
For example,  both $\cPgcd|_{\Gamma_{\Condition{x < y}{\symb{true}}}}$ and $\cPgcd|_{\Gamma_{\Condition{y < }{\symb{true}}}}$ satisfy the above assumption.

Our idea of extending the previous transformation is the use of coding;
Roughly speaking, for $\Rule{\Gamma_{T'}}{\CALL(\Gamma_{T_i},\delta_i) \ScompSym \theta_{m+i}}$ with $1 \leq i \leq n$ and for all variables $x,y$ in $T'$, we apply \emph{coding} to $\applySubst{\delta\theta_{m+i}}{x}$ and $\applySubst{\delta\theta_{m+i}}{y}$.
A variable in $\Var(T_i)$, which is instantiated by substitutions generated from $\Gamma_{T_i}$, may prevent us from constructing a finite number of production rules (see Example~\ref{ex:non-regular} below).
For this reason, we expect any variable%
\footnote{ This is not the case where either (a) or (d) holds.}
in $\Var(\applySubst{\delta\theta_{m+i}}{x},\applySubst{\delta\theta_{m+i}}{y}) \cap \Var(T_i)$ to be coded with 
\begin{itemize}
	\item $\bot$ (the case where the precondition ``$p \in \Pos(\applySubst{\delta\theta_{m+i}}{x})\cap\Pos(\applySubst{\delta\theta_{m+1}}{y})$'' does not hold), 
	\item another variable in $\Var(\applySubst{\delta\theta_{m+i}}{x},\applySubst{\delta\theta_{m+i}}{y}) \cap \Var(T_i)$ (the case where (b) above holds), or 
	\item a constructor term without any variable in $\Var(T_i)$ (the case where (c) above holds).
\end{itemize}

\begin{definition}
\label{def:SSG-conversion}	
We denote the set of constructor terms appearing in substitutions in $\cP$ by $\Patterns{\cP}$, where such constructor terms are instantiated with a non-terminal $A$ introduced during the transformation below:
$\Patterns{\cP} = \left\{ \ApplySubst{\{ x \mapsto A \mid  x \in \Var(t) \}}{t} \mid \mbox{$\theta$ appears in $\cP$}, ~ s \in \VRange(\theta), ~ t \unlhd s \right\}$.%
\footnote{
The current definition of $\Patterns{\cP}$ is not well optimized and $\Patterns{\cP}$ may include some terms that are not necessary for the transformation.
However, for readability, we adopt this simpler definition.
}
We denote the set of variables appearing in $\cN$ by $\Vars{\cN}$:
$\Vars{\cN} = \bigcup_{\Gamma_{T'} \in \cN} \Var(T')$.
The RTG obtained from $\cG$ and variables $x_1,x_2$ in $T$, denoted by $\Range(\cG,T,x_1,x_2)$,
is $(\Gamma_{T}^{(x_1,x_2)}, \cN'\cup\cN_A,\cP_1'\cup\cP_2'\cup\cP_{\NT{A}{A}}\cup\cP_{\NT{A}{\bot}}\cup\cP_{\NT{\bot}{A}})$ such that 
\begin{itemize}
	\item $\cN' = 
		\{~ \Gamma_{T'}^{(x,y)}, ~ \Gamma_{T'}^{(x,t)}, ~ \Gamma_{T'}^{(t,y)} \mid 
		x,y \in \Vars{\cN}, ~ \Gamma_{T'} \in \cN, ~ t \in \Patterns{\cP}\cup\{\bot\} ~\}
		$,
	\item $\cN_A = 
		\{~ \NT{A}{A}, ~ \NT{A}{\bot}, ~ \NT{\bot}{A} ~\}$,
	\item $\cP_1' =
		\{~ \Rule{\Gamma_{T'}^{(t_1,t_2)}}{u} \mid \Rule{\Gamma_{T'}}{\theta} \in \cP, ~ \Gamma_{T'}^{(t_1,t_2)} \in \cN', ~ \xi_A = \{ x \mapsto A \mid x \in \Var(\applySubst{\theta}{t_1},\applySubst{\theta}{t_2}) \}, ~ u \in \langle \applySubst{\xi_A\theta}{t_1},$ $ \applySubst{\xi_A\theta}{t_2} \rangle_{\top} ~\}$, 
	\item $\cP_2' =
		\{~ \Rule{\Gamma_{T'}^{(t_1,t_2)}}{u}
		\mid
		\Rule{\Gamma_{T'}}{\CALL(\Gamma_{T''},\delta) \ScompSym \theta} \in \cP, ~ \Gamma_{T'}^{(t_1,t_2)} \in \cN', \xi_A = \{ x \mapsto A \mid x \in \Var(\applySubst{\delta\theta}{t_1},\applySubst{\delta\theta}{t_2})\setminus\Var(T'') \}, ~ ~ u \in \CodingNT{\applySubst{\xi_A\delta\theta}{t_1}}{\applySubst{\xi_A\delta\theta}{t_2}}{T''}
		~\}$,
	\item $\cP_{\NT{A}{A}}=
			\{~ \Rule{\NT{A}{A}}{u} \mid \symb{f}/m,\symb{g}/n \in \cC,
			~ u \in 
			\CodingNT{
			\symb{f}(A,\ldots,A)
			}{
			\symb{g}(A,\ldots,A)
			}{\top} 
			~\}$,
	\item $\cP_{\NT{A}{\bot}}=
			\{~ \Rule{\NT{A}{\bot}}{u} \mid \symb{f}/m \in \cC,
			~ u \in 
			\CodingNT{
			\symb{f}(A,\ldots,A)
			}{\bot}{\top} ~\}$,
			and
	\item $\cP_{\NT{\bot}{A}}=
			\{~ \Rule{\NT{\bot}{A}}{u} \mid \symb{g}/n \in \cC,
			~ u \in 
			\CodingNT{\bot}{
			\symb{g}(A,\ldots,A)
			}{\top} ~\}
			$,
\end{itemize}
where $\CodingNT{\cdot}{\cdot}{T'}$, which takes a goal clause $T'$ and two terms in $\Terms(\cF\cup\{A\},\Var(T'))$ as input and returns a set of terms in $\Terms(\cF\cup\cN'\cup\cN_A)$, is recursively defined as follows:
\begin{itemize}
	\item $\CodingNT{x}{y}{T'} = \{~\Gamma_{T'}^{(x,y)}~\}$, where $x,y\in\cV$,
	\item $\CodingNT{x}{t}{T'} = \{~ \Gamma_{T'}^{(x,t)} ~\}$, where $x \in \cV$ and $t \in \Patterns{\cP}$,
	\item $\CodingNT{x}{\bot}{T'} = \{~ \Gamma_{T'}^{(x,\bot)} ~\}$, where $x \in \cV$,
	\item $\CodingNT{t}{y}{T'} = \{~ \Gamma_{T'}^{(A,y)} ~\}$, where $y \in \cV$ and $t \in \Patterns{\cP}$,
	\item $\CodingNT{\bot}{y}{T'} = \{~ \Gamma_{T'}^{(\bot,y)} ~\}$, where $y \in \cV$,
	\item $\CodingNT{A}{A}{T'} = \{~ \NT{A}{A} ~\}$,
	\item $\CodingNT{A}{\bot}{T'} = \{~ \NT{A}{\bot} ~\}$,
	\item $\CodingNT{\bot}{A}{T'} = \{~ \NT{\bot}{A} ~\}$,
	\item $\CodingNT{\bot}{\symb{g}(t_1,\ldots,t_n)}{T'} = 
		\{~ \Pairsymb{\bot}{g}(u_1,\ldots,u_n) 
		\mid 
		1 \leq i \leq n, 
		~ u_i \in \CodingNT{\bot}{t_i}{T'} 
		~\}$,
	\item $\CodingNT{\symb{f}(s_1,\ldots,s_m)}{\bot}{T'} = 
		\{~ \Pairsymb{f}{\bot}(u_1,\ldots,u_m) 
		\mid 
		1 \leq i \leq m, 
		~ u_i \in \CodingNT{s_i}{\bot}{T'} 
		~\}$,
	\item $\CodingNT{A}{\symb{g}(t_1,\ldots,t_n)}{T'} = 
		\{~ \Pairsymb{f}{g}(u_1,\ldots,u_m,u_{m+1},\ldots,u_n) 
		\mid 
		\symb{f}/m \in \cC, 
		~ m < n, 
		~ 1 \leq i \leq m, 
		$ \linebreak $u_i \in \CodingNT{A}{t_i}{T'}, 
		~ 1 \leq j \leq n-m, 
		~ u_{m+j} \in \CodingNT{\bot}{t_{m+j}}{T'}
		 ~\}
		\cup
		\{~ \Pairsymb{f}{g}(u_1,\ldots,u_n,u_{n+1},\ldots,u_m)
		\mid 
		\symb{f}/m \in \cC, 
		~ m \geq n, 
		~ 1 \leq i \leq n, 
		~ u_i \in \CodingNT{A}{t_i}{T'}, 
		~ 1 \leq j \leq m - n, 
		~ u_{n+j} \in \CodingNT{\bot}{t_{n+j}}{T'}
		~\}
		$,
	\item $\CodingNT{\symb{f}(s_1,\ldots,s_m)}{A}{T'} = 
		\{~ \Pairsymb{f}{g}(u_1,\ldots,u_m,u_{m+1},\ldots,u_n) 
		\mid 
		~ \symb{g}/n \in \cC, 
		~ m < n, 
		~ 1 \leq i \leq m, 
		~ u_i \in \CodingNT{s_i}{A}{T'}, $ $
		~ 1 \leq j \leq n-m,
		~ u_{m+j} \in \CodingNT{\bot}{A}{T'}
		~\}
		\cup
		\{~ \Pairsymb{f}{g}(u_1,\ldots,u_n,u_{n+1},\ldots,u_m) 
		\mid 
		\symb{g}/n \in \cC, 
		~ m \geq n, 
		~ 1 \leq i \leq n, 
		~ u_i \in \CodingNT{s_i}{A}{T'}, 
		~ 1 \leq j \leq m - n,
		~ u_{n+j} \in \CodingNT{s_{n+j}}{\bot}{T'}
		~\}
		$,
	\item $\CodingNT{\symb{f}(s_1,\ldots,s_m)}{\symb{g}(t_1,\ldots,t_n)}{T'} = 
		\{~ \Pairsymb{f}{g}(u_1,\ldots,u_m,u_{m+1},\ldots,u_n) 
		\mid 
		1 \leq i \leq m, 
		~ u_i \in \CodingNT{s_i}{t_i}{T'}, 
		~ 1 \leq j \leq n-m,
		~ u_{m+j} \in \CodingNT{\bot}{t_{m+j}}{T'}
		~\}$ if $m < n$, and
	\item $\CodingNT{\symb{f}(s_1,\ldots,s_m)}{\symb{g}(t_1,\ldots,t_n)}{T'} = 
		\{~ \Pairsymb{f}{g}(u_1,\ldots,u_n,u_{n+1},\ldots,u_m) 
		\mid 
		1 \leq i \leq n, 
		~ u_i \in \CodingNT{s_i}{t_i}{T'},
		~ 1 \leq j \leq m-n,
		~ u_{n+j} \in \CodingNT{s_{n+j}}{\bot}{T'}
		~\}$ if $m \geq n$.
\end{itemize}
\end{definition}
Note that the non-terminal $\NT{A}{A}$ generates $\{ \Coding{t_1}{t_2} \mid t_1,t_2 \in \Terms(\cC) \}$, the non-terminal $\NT{A}{\bot}$ generates $\{ \Coding{t_1}{\bot} \mid t_1 \in \Terms(\cC) \}$, and the non-terminal $\NT{\bot}{A}$ generates $\{ \Coding{\bot}{t_2} \mid t_2 \in \Terms(\cC) \}$.
Note also that we generate only production rules that are reachable from $\Gamma_T^{(x_1,x_2)}$, and drop from $\cN'\cup\cN_A$ non-terminals not appearing in the generated production rules.

\begin{example}
Consider 
$\cGgcd=(\Gamma_{\Condition{x < y}{\symb{true}} \Eand \Condition{y < x}{\symb{true}}}, \{\Gamma_{\Condition{x < y}{\symb{true}} \Eand \Condition{y < x}{\symb{true}}}, \Gamma_{\Condition{x < y}{\symb{true}}}, \Gamma_{\Condition{y < x}{\symb{true}}}\},\cPgcd)$ in Section~\ref{sec:intro}.
We have that 
\begin{itemize}
	\item $\Patterns{\cPgcd} = \{\symb{0},\symb{s}(A),A\}$,
		and
	\item $\Vars{\{\Gamma_{\Condition{x < y}{\symb{true}} \Eand \Condition{y < x}{\symb{true}}}, \Gamma_{\Condition{x < y}{\symb{true}}}, \Gamma_{\Condition{y < x}{\symb{true}}}\}} = \{x,y\}$.
\end{itemize}
Let us focus on $\Gamma_{\Condition{x < y}{\symb{true}}}$ and $x,y$.
Since neither $\Gamma_{\Condition{x < y}{\symb{true}} \Eand \Condition{y < x}{\symb{true}}}$ nor $\Gamma_{\Condition{y < x}{\symb{true}}}$ is reachable from $\Gamma_{\Condition{x < y}{\symb{true}}}$ by $\cPgcd$, when we construct the RTG $\Range(\cGgcd,\Gamma_{\Condition{x < y}{\symb{true}}},x,y)$, we do not take into account $\Gamma_{\Condition{x < y}{\symb{true}} \Eand \Condition{y < x}{\symb{true}}},\Gamma_{\Condition{y < x}{\symb{true}}}$, and their rules.
The RTG $\Range(\cGgcd,\Gamma_{\Condition{x < y}{\symb{true}}},x,y)=(\Gamma_{\Condition{x < y}{\symb{true}}}^{(x,y)},\cN'\cup\cN_A,\cP_1'\cup\cP_2'\cup\cP_{\NT{A}{A}}\cup\cP_{\NT{A}{\bot}}\cup\cP_{\NT{\bot}{A}})$ is constructed as follows:
\begin{itemize}
	\item $\cN'=\{ 
		\Gamma_{\Condition{x < y}{\symb{true}}}^{(x,y)}, ~ 
		\Gamma_{\Condition{x < y}{\symb{true}}}^{(y,x)}, ~ 
		\Gamma_{\Condition{y < x}{\symb{true}}}^{(x,y)} \}$,
	\item $\cN_A=\{ 
		\NT{A}{A}, ~
		\NT{A}{\bot}, ~
		\NT{\bot}{A} \}$,
	\item $\cP_1'=\{ \Rule{\Gamma_{\Condition{x < y}{\symb{true}}}^{(x,y)}}{\Pairsymb{0}{s}(\NT{\bot}{A})} \}$,
		because $\Rule{\Gamma_{\Condition{x < y}{\symb{true}}}}{\{ x\mapsto \symb{0}, ~ y \mapsto \symb{s}(y_2) \}} \in \cPgcd$ and $\CodingNT{\symb{0}}{\symb{s}(A)}{\top}$ $= \{\Pairsymb{0}{s}(\NT{\bot}{A})\}$,
	\item $\cP_2'=\{ \Rule{\Gamma_{\Condition{x < y}{\symb{true}}}^{(x,y)}}{\Pairsymb{s}{s}(\Gamma_{\Condition{x < y}{\symb{true}}}^{(x,y)})} \}$,
		because $\Rule{\Gamma_{\Condition{x < y}{\symb{true}}}}{\CALL( \Gamma_{\Condition{x < y}{\symb{true}}} , \{ x_3 \mapsto x, ~ y_3 \mapsto y \} ) \ScompSym \{ x \mapsto \symb{s}(x_3), ~ y \mapsto \symb{s}(y_3) \}} \in \cPgcd$ and $\CodingNT{\symb{s}(x)}{\symb{s}(y)}{\Condition{x < y}{\symb{true}}} = \{ \Pairsymb{s}{s}(\Gamma_{\Condition{x < y}{\symb{true}}}^{(x,y)}) \}$,
	\item $\cP_{\NT{A}{A}}= \{~
			\Rule{\NT{A}{A}}{u}
			\mid
			u \in\{
			\Pairsymb{0}{0}, ~
			\Pairsymb{0}{s}(\NT{\bot}{A}), ~
			\Pairsymb{0}{true}, ~
			\Pairsymb{0}{false}, ~
			\Pairsymb{s}{0}(\NT{A}{\bot}), ~
			\Pairsymb{s}{s}(\NT{A}{A}), ~
			\Pairsymb{s}{\,true}(\NT{A}{\bot}), ~
			\Pairsymb{s}{\,false}(\NT{A}{\bot}), ~
			\Pairsymb{true}{0},$ $
			\Pairsymb{true}{\,s}(\NT{\bot}{A}), ~
			\Pairsymb{true}{\,true}, ~
			\Pairsymb{true}{\,false}, ~
			\Pairsymb{false}{0}, ~
			\Pairsymb{false}{\,s}(\NT{\bot}{A}), ~
			\Pairsymb{false}{\,true}, ~
			\Pairsymb{false}{\,false}
			\}
			~\}$,
	\item $\cP_{\NT{A}{\bot}}= \{~
			\Rule{\NT{A}{\bot}}{u}
			\mid
			u \in\{
			\Pairsymb{0}{\bot}, ~
			\Pairsymb{s}{\bot}(\NT{A}{\bot}), ~
			\Pairsymb{true}{\bot}, ~
			\Pairsymb{false}{\bot}
			\}
			~\}$,
			and
	\item $\cP_{\NT{\bot}{A}}= \{~
			\Rule{\NT{\bot}{A}}{u}
			\mid 
			u \in\{
			\Pairsymb{\bot}{0}, ~
			\Pairsymb{\bot}{s}(\NT{\bot}{A}), ~
			\Pairsymb{\bot}{true}, ~
			\Pairsymb{\bot}{false}
			\}
			~\}$.
\end{itemize}
For $\Gamma_{\Condition{y < x}{\symb{true}}}$ and $x,y$, we add $\Rule{\Gamma_{\Condition{y < x}{\symb{true}}}^{(x,y)}}{\Gamma_{\Condition{x < y}{\symb{true}}}^{(y,x)}}$ to the above production rules.
Rules that are not reachable from $\Gamma_{\Condition{x < y}{\symb{true}}}^{(x,y)}$ or $\Gamma_{\Condition{y < x}{\symb{true}}}^{(x,y)}$ can be dropped from $\Range(\cGgcd,\Gamma_{\Condition{x < y}{\symb{true}}},x,y)$, obtaining an RTG, denoted by $\cGgcddomain$, with the following production rules:
\[
\begin{array}{@{}r@{\>}c@{\>}l@{~~~~~~~~}r@{\>}c@{\>}l@{}} 
\Rule{\Gamma_{\Condition{x < y}{\symb{true}}}^{(x,y)} &}{& \Pairsymb{0}{s}(\NT{\bot}{A}) \mid \Pairsymb{s}{s}(\Gamma_{\Condition{x < y}{\symb{true}}}^{(x,y)})}
&
\Rule{\NT{A}{\bot} &}{& \Pairsymb{0}{\bot} \mid \Pairsymb{s}{\bot}(\NT{A}{\bot}) \mid \Pairsymb{true}{\bot} \mid \Pairsymb{false}{\bot}} 
\\[5pt]
\Rule{\Gamma_{\Condition{x < y}{\symb{true}}}^{(y,x)} &}{& \Pairsymb{s}{0}(\NT{A}{\bot}) \mid \Pairsymb{s}{s}(\Gamma_{\Condition{x < y}{\symb{true}}}^{(y,x)})}
&
\Rule{\NT{\bot}{A} &}{& \Pairsymb{\bot}{0} \mid \Pairsymb{\bot}{s}(\NT{\bot}{A}) \mid \Pairsymb{\bot}{true} \mid \Pairsymb{\bot}{false}} 
\\[5pt]	
\Rule{\Gamma_{\Condition{y < x}{\symb{true}}}^{(x,y)} &}{& \Gamma_{\Condition{x < y}{\symb{true}}}^{(y,x)}} 
\\[5pt]
\end{array}
\]	
It is easy to see that
\begin{itemize}
	\item $\Lang(\cGgcddomain,\Gamma_{\Condition{x < y}{\symb{true}}}^{(x,y)}) \subseteq \Terms(\{\Pairsymb{0}{s},\Pairsymb{s}{s},\Pairsymb{\bot}{0},\Pairsymb{\bot}{s},\Pairsymb{\bot}{true},\Pairsymb{\bot}{false}\})$, 
	\item $\Lang(\cGgcddomain,\Gamma_{\Condition{y < x}{\symb{true}}}^{(x,y)}) \subseteq \Terms(\{\Pairsymb{s}{0},\Pairsymb{s}{s},\Pairsymb{0}{\bot},\Pairsymb{s}{\bot},\Pairsymb{true}{\bot},\Pairsymb{false}{\bot}\})$,
\end{itemize}
and hence, there is no shared constant between the two sets.
This means that 
\[
\Lang(\cGgcddomain,\Gamma_{\Condition{x < y}{\symb{true}}}^{(x,y)}) \cap \Lang(\cGgcddomain,\Gamma_{\Condition{y < x}{\symb{true}}}^{(x,y)}) = \emptyset
\]
and hence
\[
\SubstSet{\cGgcd}{\Gamma_{\Condition{x < y}{\symb{true}}}} \cap \SubstSet{\cGgcd}{\Gamma_{\Condition{y < x}{\symb{true}}}} = \emptyset.
\]
Note that the emptiness problem of RTGs is decidable, and hence we can decide the emptiness problem of $\Lang(\cGgcddomain,\Gamma_{\Condition{x < y}{\symb{true}}}^{(x,y)}) \cap \Lang(\cGgcddomain,\Gamma_{\Condition{y < x}{\symb{true}}}^{(x,y)})$.
\end{example}

The following example illustrates both why not all SSGs can be transformed and why we adopt the assumption.
\begin{example}
\label{ex:non-regular}
Let $\cGcex$ be the following SSG which does not satisfy the assumption:
\[
(\Gamma_{\Condition{x}{y}},\{\Gamma_{\Condition{x}{y}}\},
\{~
\Rule{\Gamma_{\Condition{x}{y}}}{
\{x\mapsto\symb{0}, ~ y \mapsto\symb{0}\}
\mid
\CALL(\Gamma_{\Condition{x}{y}},\{x'\mapsto x,~y'\mapsto y\}) \ScompSym \{ x \mapsto \symb{s}(x'), ~ y \mapsto \symb{s}(\symb{s}(y')) \}
}
~\}).
\]
The domains of substitutions generated by $\cGcex$ w.r.t.\ $x,y$ is 
$\{ (\symb{s}^n(\symb{0}),\symb{s}^{2n}(\symb{0})) \mid n \geq 0 \}$ which is not recognizable. 
This implies that there is no RTG generating this set, while every substitution appearing in $\cGcex$ preserves linearity.

Let us now apply $\Range$ to $\cGcex$, while $\cGcex$ does not satisfy the assumption.
To generate rules from $\Rule{\Gamma_{\Condition{x}{y}}}{
\CALL(\Gamma_{\Condition{x}{y}},\{x'\mapsto x,~y'\mapsto y\}) \ScompSym \{ x \mapsto \symb{s}(x'), ~ y \mapsto \symb{s}(\symb{s}(y')) \}
}$, we need to compute $\CodingNT{\symb{s}(x)}{\symb{s}(\symb{s}(y))}{\Gamma_{\Condition{x}{y}}}$, resulting in $\Pairsymb{s}{s}(\CodingNT{x}{\symb{s}(y)}{\Gamma_{\Condition{x}{y}}})$.
The first argument $x$ of $\CodingNT{x}{\symb{s}(y)}{\Gamma_{\Condition{x}{y}}}$ cannot be instantiated any more without $\Gamma_{\Condition{x}{y}}$.
Then, let us define $\CodingNT{x}{\symb{s}(y)}{\Gamma_{\Condition{x}{y}}}=\Gamma_{\Condition{x}{y}}^{(x,\symb{s}(y))}$.
Then, the non-terminal $\Gamma_{\Condition{x}{y}}^{(x,\symb{s}(y))}$ is not generated in computing the set of non-terminals ($\cN'\cup\cN_A$ in Definition~\ref{def:SSG-conversion}).
Let us now add $\Gamma_{\Condition{x}{y}}^{(x,\symb{s}(y))}$ into the set of non-terminals, and generate rules for $\Gamma_{\Condition{x}{y}}^{(x,\symb{s}(y))}$ from $\Rule{\Gamma_{\Condition{x}{y}}}{
\CALL(\Gamma_{\Condition{x}{y}},\{x'\mapsto x,~y'\mapsto y\}) \ScompSym \{ x \mapsto \symb{s}(x'), ~ y \mapsto \symb{s}(\symb{s}(y')) \}
}$.
Then, we need non-terminal $\Gamma_{\Condition{x}{y}}^{(x,\symb{s}(\symb{s}(\symb{s}(y))))}$.
In summary, we need infinitely many non-terminals and their production rules.
The assumption enables us to avoid such a case.
\end{example}

Finally, we show correctness of the transformation in Definition~\ref{def:SSG-conversion}, i.e., that
$\Lang(\Range(\cG,T,x_1,x_2))$ is an overapproximation of the ranges of ground substitutions obtained from $\SubstSet{\cG}{\Gamma_T}$ w.r.t.\ $x_1,x_2$.
We first show some auxiliary lemmas, and then show the main theorem.
\begin{lemma}
\label{lem:CodingNT-subst}
Let $T$ be a goal clause, $t_1,t_2 \in \Terms(\cC,\cV)$, $\theta \in \Subst(\cC)$, $\xi \in \Subst(\cC)$ such that 
	$\Dom(\theta)\cap\Dom(\xi)=\emptyset$ and
	$\Dom(\theta)\cup\Dom(\xi) = \Var(t_1,t_2)$.
Note that $\theta\cup\xi=\theta\xi=\xi\theta$.
Let $\xi_A = \{ x \mapsto A \mid x \in \Dom(\xi)\}$ and $u \in \CodingNT{\applySubst{\xi_A}{t_1}}{\applySubst{\xi_A}{t_2}}{T}$.
Suppose that for all positions $p \in 
\Pos(t_1)\cap\Pos(t_2)$,
	both of the following hold:
	\begin{itemize}
		\item if $t_1|_p \in \Dom(\theta)$, then $t_2|_p \in \Dom(\theta) \cup \Terms(\cC,\Dom(\xi))$,
		and
		\item if $t_2|_p \in \Dom(\theta)$, then $t_1|_p \in \Dom(\theta) \cup \Terms(\cC,\Dom(\xi))$.
	\end{itemize}
Then, all of the following hold:
\begin{enumerate}
\renewcommand{\labelenumi}{(\alph{enumi})}
	\item $\Pos(\Coding{\applySubst{\theta\xi}{t_1}}{\applySubst{\theta\xi}{t_2}}) \supseteq \Pos(\Coding{t_1}{t_2}) \supseteq \Pos(u)$
		(i.e., $\Pos(t_1) \cup \Pos(t_2) \supseteq \Pos(u)$),
	\item for any position $p \in \Pos(t_1)\cap\Pos(t_2)$, all of the following hold:
	\begin{itemize}
		\item if $t_1|_p=x\in \Dom(\theta)$ and $t_2|_p=y \in \Dom(\theta)$, then
			$(\Coding{t_1}{t_2})|_p = \Pairsymb{\var{x}}{\var{y}}$ (i.e., $(\Coding{\applySubst{\theta\xi}{t_1}}{\applySubst{\theta\xi}{t_2}})|_p$ $= \Coding{\applySubst{\theta}{x}}{\applySubst{\theta}{y}})$) and 
			$u|_p=\Gamma_T^{(x,y)}$
		\item if $t_1|_p=x\in \Dom(\theta)$ and $t_2|_p=y \in \Dom(\xi)$, then
			$(\Coding{t_1}{t_2})|_p = \Pairsymb{\var{x}}{\var{y}}$ (i.e., $(\Coding{\applySubst{\theta\xi}{t_1}}{\applySubst{\theta\xi}{t_2}})|_p$ $= \Coding{\applySubst{\theta}{x}}{\applySubst{\xi}{y}}$) and
			$u|_p=\Gamma_T^{(x,A)}$
		\item if $t_1|_p=x\in \Dom(\xi)$ and $t_2|_p=y \in \Dom(\theta)$, then
			$(\Coding{t_1}{t_2})|_p = \Pairsymb{\var{x}}{\var{y}}$ (i.e, $(\Coding{\applySubst{\theta\xi}{t_1}}{\applySubst{\theta\xi}{t_2}})|_p$ $= \Coding{\applySubst{\xi}{x}}{\applySubst{\theta}{y}}$) and
			$u|_p=\Gamma_T^{(A,y)}$
		\item if $t_1|_p=x\in \Dom(\xi)$ and $t_2|_p=y \in \Dom(\xi)$, then
			$(\Coding{t_1}{t_2})|_p = \Pairsymb{\var{x}}{\var{y}}$ (i.e., $(\Coding{\applySubst{\theta\xi}{t_1}}{\applySubst{\theta\xi}{t_2}})|_p$ $= \Coding{\applySubst{\xi}{x}}{\applySubst{\xi}{y}}$) and
			$u|_p=\NT{A}{A}$
		\item if $t_1|_p=x\in \Dom(\theta)$ and $\Root(t_2|_p)=\symb{g} \in \cC$, then 
			$\Root((\Coding{t_1}{t_2})|_p) = \Pairsymb{\var{x}}{g}$ (i.e., $(\Coding{\applySubst{\theta\xi}{t_1}}{\applySubst{\theta\xi}{t_2}})|_p$ $= \Coding{\applySubst{\theta}{x}}{\ApplySubst{\xi}{t_2|_p}}$) and
			$u|_p=\Gamma_T^{(x,\ApplySubst{\xi_A}{t_2|_p})}$
		\item if $t_1|_p=x\in \Dom(\xi)$ and $\Root(t_2|_p)=\symb{g} \in \cC$, then 
			$\Root((\Coding{t_1}{t_2})|_p) = \Pairsymb{\var{x}}{g}$ (i.e., $(\Coding{\applySubst{\theta\xi}{t_1}}{\applySubst{\theta\xi}{t_2}})|_p$ $= \Coding{\applySubst{\xi}{x}}{\ApplySubst{\xi}{t_2|_p}}$),
			$u|_p \in \CodingNT{A}{\ApplySubst{\xi_A}{t_2|_p}}{\bot}$,
			and
			there exists a term $t'_1 \in \Terms(\cC,\cV)$ and a term $u' \in \CodingNT{A}{\ApplySubst{\xi_A}{t_2|_p}}{\bot}$ such that
				$t'_1 \leq \applySubst{\xi}{x}$,
				$u' = \Coding{\applySubst{\xi'_A}{t'_1}}{t_2|_p}$,
			\begin{itemize}
				\item 
				for all $q \in \Pos(t'_1)\cap\Pos(t_2|_p)$, $\ApplySubst{\xi'_A}{t'_1|_q}=A$ if and only if $t_2|_{pq}=A$,%
					\footnote{
					This implies that if $q \in \Pos(t'_1)\cap\Pos(t_2|_p)$, then $pq \in \Pos(u')$ and $u|_{pq}= \NT{A}{A}$.
					}
				and
				\item 
				for all $q \in \Pos(t'_1)\setminus\Pos(t_2|_p)$, $\ApplySubst{\xi'_A}{t'_1|_q}=A$,%
					\footnote{
					This implies that if $q \in \Pos(t'_1)\setminus\Pos(t_2|_p)$, then $pq \in \Pos(u')$ and $u|_{pq}= \NT{A}{\bot}$.
					}
			\end{itemize}
			where $\xi'_A=\{x\mapsto A \mid x \in \Dom(t'_1) \}$,
		\item if $\Root(t_1|_p)=\symb{f} \in \cC$ and $t_2|_p=y \in \Dom(\theta)$, then 
			$\Root((\Coding{t_1}{t_2})|_p) = \Pairsymb{f}{\var{y}}$ (i.e., $(\Coding{\applySubst{\theta\xi}{t_1}}{\applySubst{\theta\xi}{t_2}})|_p$ $= \Coding{\ApplySubst{\xi}{t_1|_p}}{\applySubst{\theta}{y}}$) and
			$u|_p=\Gamma_T^{(\ApplySubst{\xi_A}{t_1|_p},y)}$
		\item if $\Root(t_1|_p)=\symb{f} \in \cC$ and $t_2|_p=y \in \Dom(\xi)$, then 
			$\Root((\Coding{t_1}{t_2})|_p) = \Pairsymb{f}{\var{y}}$ (i.e., $(\Coding{\applySubst{\theta\xi}{t_1}}{\applySubst{\theta\xi}{t_2}})|_p$ $= \Coding{\ApplySubst{\xi}{t_1|_p}}{\applySubst{\xi}{y}}$),
			$u|_p \in \CodingNT{\ApplySubst{\xi_A}{t_1|_p}}{A}{\top}$,
			and
			there exists a term $t'_2 \in \Terms(\cC,\cV)$ and a term $u' \in \CodingNT{\ApplySubst{\xi_A}{t_1|_p}}{A}{\top}$ such that
				$t'_2 \leq \applySubst{\xi}{y}$,
				$u' = \Coding{t_1|_p}{\applySubst{\xi'_A}{t'_2}}$,
			\begin{itemize}
				\item 
				for all $q \in \Pos(t'_2)\cap\Pos(t_1|_p)$, $\ApplySubst{\xi'_A}{t'_2|_q}=A$ if and only if $t_1|_{pq}=A$,
				and
				\item 
				for all $q \in \Pos(t'_2)\setminus\Pos(t_1|_p)$, $\ApplySubst{\xi'_A}{t'_2|_q}=A$,
			\end{itemize}
			where $\xi'_A=\{x\mapsto A \mid x \in \Dom(t'_2) \}$,
		\item if $\Root(t_1|_p)=\symb{f}\in \cC$ and $\Root(t_2|_p)=\symb{g} \in \cC$, then 
			$\Root((\Coding{t_1}{t_2})|_p) = \Root((\Coding{\applySubst{\theta\xi}{t_1}}{\applySubst{\theta\xi}{t_2}})|_p) = \Root(u|_p) = \Pairsymb{f}{g}$ (i.e., $(\Coding{\applySubst{\theta\xi}{t_1}}{\applySubst{\theta\xi}{t_2}})|_p = \Coding{\ApplySubst{\theta\xi}{t_1|_p}}{\ApplySubst{\theta\xi}{t_2|_p}}$),
	\end{itemize}
	\item for any position $p \in \Pos(t_1)\setminus\Pos(t_2)$, both of the following hold:
	\begin{itemize}
		\item if $t_1|_p = x \in \Dom(\theta)$, then 
			$(\Coding{t_1}{t_2})|_p = \Pairsymb{\var{x}}{\bot}$ (i.e, $(\Coding{\applySubst{\theta\xi}{t_1}}{\applySubst{\theta\xi}{t_2}})|_p=(\Coding{\applySubst{\theta}{x}}{\bot})|_p$) and
			$u|_p=\Gamma_T^{(x,\bot)}$,
		\item if $t_1|_p = x \in \Dom(\xi)$, then 
			$(\Coding{t_1}{t_2})|_p = \Pairsymb{\var{x}}{\bot}$ (i.e., $(\Coding{\applySubst{\theta\xi}{t_1}}{\applySubst{\theta\xi}{t_2}})|_p=(\Coding{\applySubst{\xi}{x}}{\bot})|_p$) and
			$u|_p=\NT{A}{\bot}$,
		and
		\item if $\Root(t_1|_p) = \symb{f} \in \cC$, then 
			$\Root((\Coding{t_1}{t_2})|_p) = \Root((\Coding{\applySubst{\theta\xi}{t_1}}{\applySubst{\theta\xi}{t_2}})|_p) = \Root(u|_p) = \Pairsymb{f}{\bot}$ (i.e., $(\Coding{\applySubst{\theta\xi}{t_1}}{\applySubst{\theta\xi}{t_2}})|_p=(\Coding{\ApplySubst{\theta\xi}{t_1|_p}}{\bot})|_p$),
	\end{itemize}
	and
	\item for any position $p \in \Pos(t_2)\setminus\Pos(t_1)$, both of the following hold:
	\begin{itemize}
		\item if $t_2|_p = y \in \Dom(\theta)$, then 
			$(\Coding{t_1}{t_2})|_p = \Pairsymb{\bot}{\var{y}}$ (i.e., $(\Coding{\applySubst{\theta\xi}{t_1}}{\applySubst{\theta\xi}{t_2}})|_p=(\Coding{\bot}{\applySubst{\theta}{y}})|_p$) and
			$u|_p=\Gamma_T^{(\bot,y)}$,
		\item if $t_2|_p = y \in \Dom(\xi)$, then 
			$(\Coding{t_1}{t_2})|_p = \Pairsymb{\bot}{\var{y}}$ (i.e., $(\Coding{\applySubst{\theta\xi}{t_1}}{\applySubst{\theta\xi}{t_2}})|_p=(\Coding{\bot}{\applySubst{\xi}{y}})|_p$) and
			$u|_p=\NT{\bot}{A}$,
		and
		\item if $\Root(t_2|_p) = \symb{g} \in \cC$, then
			$\Root((\Coding{t_1}{t_2})|_p) = \Root((\Coding{\applySubst{\theta\xi}{t_1}}{\applySubst{\theta\xi}{t_2}})|_p) = \Root(u|_p) = \Pairsymb{\bot}{g}$ (i.e., $(\Coding{\applySubst{\theta\xi}{t_1}}{\applySubst{\theta\xi}{t_2}})|_p=(\Coding{\bot}{\ApplySubst{\theta\xi}{t_2|_p}})|_p$).
	\end{itemize}
\end{enumerate}
\end{lemma}
\begin{proof}
By definition, the claim (a) is trivial.
The claims (b)--(d) can be proved by induction on the length of $p$.
For the claims (b), (c), and (d), we make a case distinction depending on
what $t_1|_p$ and $t_2|_p$ are,
what $t_1|_p$ is, and
what $t_2|_p$ is, respectively.
\qed
\end{proof}

\begin{lemma}
\label{lem:CodingNT-context}
Let $T$ be a goal clause, $t_1,t_2 \in \Terms(\cC,\cV)$, $\theta \in \Subst(\cC)$, $\xi \in \Subst(\cC)$ such that 
	$\Dom(\theta)\cap\Dom(\xi)=\emptyset$ and
	$\Dom(\theta)\cup\Dom(\xi) = \Var(t_1,t_2)$.
Note that $\theta\cup\xi=\theta\xi=\xi\theta$.
Let $\xi_A = \{ x \mapsto A \mid x \in \Dom(\xi)\}$.
Suppose that for all positions $p \in 
\Pos(t_1)\cap\Pos(t_2)$,
	both of the following hold:
	\begin{itemize}
		\item if $t_1|_p \in \cV\cap\Dom(\theta)$, then $t_2|_p \in \Dom(\theta) \cup \Terms(\cC,\Dom(\xi))$,
		and
		\item if $t_2|_p \in \cV\cap\Dom(\theta)$, then $t_1|_p \in \Dom(\theta) \cup \Terms(\cC,\Dom(\xi))$.
	\end{itemize}
Then, there exists a term $u \in \CodingNT{\applySubst{\xi_A}{t_1}}{\applySubst{\xi_A}{t_2}}{T}$, a context $C[\,] \in \Terms((\cC\cup\{\bot\})^2\cup\{\Hole\})$, terms $S_1,\ldots,S_n$, and non-terminals $\Gamma_1,\ldots,\Gamma_n$ such that
	$\Coding{\applySubst{\theta\xi}{t_1}}{\applySubst{\theta\xi}{t_2}} = C[S_1,\ldots,S_n]$,
	$u=C[\Gamma_1,\ldots,\Gamma_n]$,
	and
	for all $1 \leq i \leq n$, all of the following hold:
	\begin{itemize}
		\item $S_i=\Pairsymb{\var{x}}{\var{y}}$ if and only if\/ $\Gamma_i=\Gamma_T^{(x,y)}$,
		\item $S_i=\Coding{\applySubst{\theta}{x}}{\ApplySubst{\xi}{t_2|_p}}$ if and only if\/ $\Gamma_i=\Gamma_T^{(x,\ApplySubst{\xi_A}{t_2|_p})}$ for some $p \in \Pos(t_2)$,
		\item $S_i=\Coding{\applySubst{\theta}{x}}{\bot}$ if and only if\/ $\Gamma_i=\Gamma_T^{(x,\bot)}$,
		\item $S_i=\Coding{\ApplySubst{\xi}{t_1|_p}}{\applySubst{\theta}{y}}$ if and only if\/ $\Gamma_i=\Gamma_T^{(\ApplySubst{\xi_A}{t_1|_p},y)}$ for some $p \in \Pos(t_1)$,
		\item $S_i=\Coding{\bot}{\applySubst{\theta}{y}}$ if and only if\/ $\Gamma_i=\Gamma_T^{(\bot,y)}$,
		\item $S_i=\Coding{\ApplySubst{\xi}{t_1|_p}}{\ApplySubst{\xi}{t_2|_p}}$ for some $p \in \Pos(\applySubst{\xi}{t_1})\cap\Pos(\applySubst{\xi}{t_2})$ if and only if\/ $\Gamma_i=\NT{A}{A}$,
		\item $S_i=\Coding{\ApplySubst{\xi}{t_1|_p}}{\bot}$ for some $p \in \Pos(\applySubst{\xi}{t_1})\setminus\Pos(\applySubst{\xi}{t_2})$ if and only if\/ $\Gamma_i=\NT{A}{\bot}$,
		and
		\item $S_i=\Coding{\bot}{\ApplySubst{\xi}{t_2|_p}}$ for some $p \in \Pos(\applySubst{\xi}{t_2})\setminus\Pos(\applySubst{\xi}{t_1})$ if and only if\/ $\Gamma_i=\NT{\bot}{A}$.
	\end{itemize}
\end{lemma}
\begin{proof}
Using Lemma~\ref{lem:CodingNT-subst}, this lemma can be proved by structural induction on $t_1,t_2$.
\qed	
\end{proof}
\begin{lemma}
\label{lem:CodingNT-no-recursion}
Let $\cG$ be an SSG $(\Gamma_{T_0},\cN,\cP)$, $\Gamma_T \in \cN$, $x_1,x_2 \in \Var(T)$, $\Range(\cG,T,x_1,x_2)$ be constructed, and
$\cG'=\Range(\cG,T,x_1,x_2)$.
Let $t_1,t_2 \in \Terms(\cC,\cV)$, $\xi \in \Subst(\cC)$ with $\Dom(\xi) \supseteq \Var(t_1,t_2)$, and $\xi_A = \{ x \mapsto A \mid x \in \Var(t_1,t_2) \}$.
Then, all of the following hold:
\begin{itemize}
	\item there exists a term $u \in \CodingNT{\applySubst{\xi_A}{t_1}}{\applySubst{\xi_A}{t_2}}{\top}$ such that $u \mathrel{\to^*_{\cG'}} \Coding{\applySubst{\xi}{t_1}}{\applySubst{\xi}{t_2}}$,
	\item there exists a term $u \in \CodingNT{\applySubst{\xi_A}{t_1}}{\bot}{\top}$ such that $u \mathrel{\to^*_{\cG'}} \Coding{\applySubst{\xi}{t_1}}{\bot}$,
	and
	\item there exists a term $u \in \CodingNT{\bot}{\applySubst{\xi_A}{t_2}}{\top}$ such that $u \mathrel{\to^*_{\cG'}} \Coding{\bot}{\applySubst{\xi}{t_2}}$.
\end{itemize}
\end{lemma}
\begin{proof}
Using the definition of $\cP_{\NT{A}{A}}$, $\cP_{\NT{A}{\bot}}$, and $\cP_{\NT{\bot}{A}}$, and Lemma~\ref{lem:CodingNT-subst}, this lemma can be proved by structural induction on $t_1,t_2$.
\qed	
\end{proof}

\begin{theorem}
\label{thm:RTG-simplification}
Let $\cG$ be an SSG $(\Gamma_{T_0},\cN,\cP)$, $\Gamma_T \in \cN$, $x_1,x_2 \in \Var(T)$, and $\Range(\cG,T,x_1,x_2)$ be constructed (i.e., $\cP|_T$ satisfies the assumption).
Then, 
\[
\Lang(\Range(\cG,T,x_1,x_2)) \supseteq \{ \Coding{\applySubst{\xi\theta}{x_1}}{\applySubst{\xi\theta}{x_2}} \mid \theta \in \SubstSet{\cG}{\Gamma_T}, ~ \xi \in \Subst(\cC), ~ \Var(\applySubst{\theta}{x_1},\applySubst{\theta}{x_2}) = \Dom(\xi) \}.
\]
\end{theorem}
\begin{proof}
Let $\cG'=\Range(\cG,T,x_1,x_2)$.
It suffices to show that
for all $\Gamma_{T'} \in \cN$, $t_1,t_2 \in \Var(T')\cup\Patterns{\cP}\cup\{\bot\}$ with $\{t_1,t_2\}\cap\cV \ne \emptyset$, and $e \in \Lang(\cG,\Gamma_{T'})$ with
$\theta=\Semantics{e}{}$, we have $\Coding{\applySubst{\xi\theta}{t_1}}{\applySubst{\xi\theta}{t_2}} \in \Lang(\cG',\Gamma_{T'}^{(t_1,t_2)})$ for all substitutions $\xi \in \Subst(\cC)$ with $\Dom(\xi)=\Var(\applySubst{\xi\theta}{t_1},\applySubst{\xi\theta}{t_2})$.
We prove this claim by induction on the length of derivations from $\Gamma_{T'}$ to $e$.
We make a case distinction depending on which rule is applied at the first step.
\begin{itemize}
	\item The case where $\Rule{\Gamma_{T'}}{\theta}$ is applied.
By construction, we have the following production rule
$\Rule{\Gamma_{T'}^{(t_1,t_2)}}{u} \in \cG'$ for each $u \in \CodingNT{\applySubst{\xi_A\theta}{t_1}}{\applySubst{\xi_A\theta}{t_2}}{\top}$,
where $\xi_A = \{ x \mapsto A \mid x \in \Var(\applySubst{\xi_A\theta}{t_1},\applySubst{\xi_A\theta}{t_2}) \}$.
Then, the claim follows from Lemma~\ref{lem:CodingNT-no-recursion}.

	\item The remaining case where $\Rule{\Gamma_{T'}}{\CALL(\Gamma_{T''},\delta) \ScompSym \sigma}$ is applied.
		Suppose that $\Gamma_{T''} \mathrel{\to^*_{\cG'}} e'$ and $\theta = \Semantics{\CALL(e',\delta)\ScompSym \sigma}{}$.
		Let $\theta'=\Semantics{e'}{}$.
		Then, $\theta=(\theta'\delta)|_{\Dom(\delta)}\sigma$.
		By construction, we have the following production rule $\Rule{\Gamma_{T'}^{(t_1,t_2)}}{u} \in \cG'$ for each $u \in \CodingNT{\applySubst{\xi_A\delta\sigma}{t_1}}{\applySubst{\xi_A\delta\sigma}{t_2}}{T'}$.
		where $\xi_A = \{ x \mapsto A \mid x \in \Var(\applySubst{\delta\sigma}{t_1},\applySubst{\delta\sigma}{t_2})\setminus\Var(T'') \}$.
		By the assumption, $\Var(T'')=\VRange(\delta)$, and thus, $\Dom(\xi_A) \cap \VRange(\delta)=\emptyset$.
		Since $\xi_A$ is a ground substitution, we have that
		$
		\applySubst{\xi_A\delta\sigma}{t_i}
		=
		\applySubst{\delta\xi_A|_{\VRange(\sigma)\setminus\Dom(\delta)}\sigma}{t_i}
		$.
		Since $\xi$ is a ground substitution and $\Dom(\xi)\supseteq \VRange(\theta)=\VRange((\theta'\delta)|_{\Dom(\delta)}\sigma)$, we have that
		\begin{itemize}
			\item 
		$
		\applySubst{\xi\theta}{t_i}=\applySubst{\xi(\theta'\delta)|_{\Dom(\delta)}\sigma}{t_i}
		=
		\applySubst{\xi\theta'\delta\xi|_{\Dom(\VRange(\sigma)\setminus \Dom(\delta)}\sigma}{t_i}
		$,
		and
		\item 
		$
		\applySubst{\xi_A\delta\sigma}{t_i}
		=
		\applySubst{\xi_A\delta\xi_A|_{\Dom(\VRange(\sigma)\setminus \Dom(\delta)}\sigma}{t_i}
		$, 
		\end{itemize}
		and hence
		\[
		\CodingNT{\applySubst{\xi_A\delta\sigma}{t_1}}{\applySubst{\xi_A\delta\sigma}{t_2}}{T''}
		=
		\CodingNT{\applySubst{\xi_A\delta\xi_A|_{\Dom(\VRange(\sigma)\setminus \Dom(\delta)}\sigma}{t_1}}{\applySubst{\xi_A\delta\xi_A|_{\Dom(\VRange(\sigma)\setminus \Dom(\delta)}\sigma}{t_2}}{T''}
		.\]
		It follows from Lemma~\ref{lem:CodingNT-context} that there exists a term $u \in \CodingNT{\applySubst{\xi_A\delta\sigma}{t_1}}{\applySubst{\xi_A\delta\sigma}{t_2}}{T''}$, a context $C[\,] \in \Terms((\cC\cup\{\bot\})^2\cup\{\Hole\})$, terms $S_1,\ldots,S_n$, and non-terminals $\Gamma_1,\ldots,\Gamma_n$ such that
		$\Coding{\applySubst{\xi\theta}{t_1}}{\applySubst{\xi\theta}{t_2}} 
		= C[S_1,\ldots,S_n]$,
		$u=C[\Gamma_1,\ldots,\Gamma_n]$,
		and
	for all $1 \leq i \leq n$, all of the following hold:
	\begin{itemize}
		\item $S_i=\Pairsymb{\var{x}}{\var{y}}$ if and only if\/ $\Gamma_i=\Gamma_{T''}^{(x,y)}$,
		\item $S_i=\Coding{\applySubst{\xi\theta}{x}}{\ApplySubst{\xi\theta}{t_2|_p}}$ if and only if\/ $\Gamma_i=\Gamma_{T''}^{(x,\ApplySubst{\xi_A}{t_2|_p})}$ for some $p \in \Pos(t_2)$,
		\item $S_i=\Coding{\applySubst{\xi\theta}{x}}{\bot}$ if and only if\/ $\Gamma_i=\Gamma_{T''}^{(x,\bot)}$,
		\item $S_i=\Coding{\ApplySubst{\xi\theta}{t_1|_p}}{\applySubst{\xi\theta}{y}}$ if and only if\/ $\Gamma_i=\Gamma_{T''}^{(\ApplySubst{\xi_A}{t_1|_p},y)}$ for some $p \in \Pos(t_1)$,
		\item $S_i=\Coding{\bot}{\applySubst{\xi\theta}{y}}$ if and only if\/ $\Gamma_i=\Gamma_{T''}^{(\bot,y)}$,
		\item $S_i=\Coding{\ApplySubst{\xi\theta}{t_1|_p}}{\ApplySubst{\xi\theta}{t_2|_p}}$ for some $p \in \Pos(\applySubst{\xi\theta}{t_1})\cap\Pos(\applySubst{\xi\theta}{t_2})$ if and only if\/ $\Gamma_i=\NT{A}{A}$,
		\item $S_i=\Coding{\ApplySubst{\xi\theta}{t_1|_p}}{\bot}$ for some $p \in \Pos(\applySubst{\xi\theta}{t_1})\setminus\Pos(\applySubst{\xi\theta}{t_2})$ if and only if\/ $\Gamma_i=\NT{A}{\bot}$,
		and
		\item $S_i=\Coding{\bot}{\ApplySubst{\xi\theta}{t_2|_p}}$ for some $p \in \Pos(\applySubst{\xi\theta}{t_2})\setminus\Pos(\applySubst{\xi\theta}{t_1})$ if and only if\/ $\Gamma_i=\NT{\bot}{A}$.
	\end{itemize}
In the case where $\Gamma_i$ is $\Gamma_{T''}^{(x,y)}$, $\Gamma_{T''}^{(x,\ApplySubst{\xi_A}{t_2|_p})}$, $\Gamma_{T''}^{(x,\bot)}$, $\Gamma_{T''}^{(\ApplySubst{\xi_A}{t_1|_p},y)}$, or $\Gamma_{T''}^{(\bot,y)}$, it follows from the induction hypothesis that $\Gamma_i \mathrel{\to^*_{\cG'}} S_i$.
In the remaining case where $\Gamma_i$ is $\NT{A}{A}$, $\NT{A}{\bot}$, or $\NT{\bot}{A}$, it follows from Lemma~\ref{lem:CodingNT-no-recursion} that $\Gamma_i \mathrel{\to^*_{\cG'}} S_i$.
Therefore, we have that
$
	\Gamma_{T'}^{(t_1,t_2)} \mathrel{\to_{\cG'}} u 
	= C[\Gamma_1,\ldots,\Gamma_n]
	\mathrel{\to^*_{\cG'}}
	C[S_1,\ldots,S_n] = \Coding{\applySubst{\xi\theta}{t_1}}{\applySubst{\xi\theta}{t_2}}
$,
	and hence, $\Coding{\applySubst{\xi\theta}{t_1}}{\applySubst{\xi\theta}{t_2}} \in \Lang(\cG',\Gamma_{T'}^{(t_1,t_2)})$.
\qed
\end{itemize}
\end{proof}
The converse inclusion (i.e., $\Lang(\Range(\cG,T,x_1,x_2)) \subseteq \{ \Coding{\applySubst{\xi\theta}{x_1}}{\applySubst{\xi\theta}{x_2}} \mid \theta \in \SubstSet{\cG}{\Gamma_T}, ~
	\ldots
 \}$)
does not hold in general (cf.~\cite[Example~31]{NM18fscd}).

\section{Conclusion}
\label{sec:conclusion}

In this paper, under a certain syntactic condition, we showed a transformation of the grammar representation of a narrowing tree into an RTG that overapproximately generates the ranges of ground substitutions generated by the grammar representation.
We showed a precise definition of the transformation and proved that the language of the transformed RTG is an overapproximation of the ranges of ground substitutions generated by the grammar representation.
We will make an experiment to evaluate the usefulness of the transformation in e.g., proving confluence of CTRSs.

The syntactic assumption in Section~\ref{sec:transformation} is a sufficient condition to, given an SSG, obtain an RTG that generates the ranges of ground substitutions generated by the SSG.
It is not known yet whether the assumption is a necessary condition or not.
We will try to clarify this point.

As stated in Section~\ref{sec:transformation},  the converse inclusion of Theorem~\ref{thm:RTG-simplification}, $\Lang(\Range(\cG,T,x_1,x_2)) \subseteq \{ \Coding{\applySubst{\xi\theta}{x_1}}{\applySubst{\xi\theta}{x_2}} \mid \theta \in \SubstSet{\cG}{\Gamma_T}, ~
	\ldots
 \}$, does not hold in general.
However, the converse inclusion must hold for an SSG such that all substitutions in the SSG preserve linearity, i.e.,
for any substitution $\sigma$ in the SSG, $\applySubst{\sigma}{x}$ is linear for all $x \in \Dom(\sigma)$, and $\Var(\applySubst{\sigma}{x})\cap\Var(\applySubst{\sigma}{y})=\emptyset$ for all $x,y \in \Dom(\sigma)$ such that $x\ne y$.
We will prove this conjecture and try to find other sufficient conditions for the converse inclusion.

\paragraph{Acknowledgements}
We gratefully acknowledge the anonymous reviewers for their useful comments and suggestions to improve the paper.



\end{document}